\documentclass{lmcs} %%% last changed 2014-08-20
\pdfoutput=1
\usepackage[utf8]{inputenc}

\usepackage{lastpage}
\lmcsdoi{21}{4}{22}
\lmcsheading{}{\pageref{LastPage}}{}{}%
{Feb.~14,~2022}{Nov.~11,~2025}{}

\keywords{model checking,
quantitative reasoning,
formal logic,
quantum computing}

\usepackage{amsmath,amssymb}
\usepackage[noend]{algpseudocode}
\usepackage{tikz}
\usetikzlibrary{shapes,arrows,automata}
\usepackage{textcomp}
\usepackage[hidelinks]{hyperref}
\usepackage{braket}
\usepackage[qm]{qcircuit}
\usepackage{dsfont}
\usepackage{lineno}
\usepackage{subcaption}
\usepackage{booktabs}
\usepackage{multirow}
\usepackage{hyperref}
\theoremstyle{plain} %\crefname{satz}{Satz}{S\"atze}
\newcommand{\QC}{\mathfrak{Q}}
\newcommand{\h}{\mathcal{H}}
\renewcommand{\L}{\mathcal{L}}
\newcommand{\D}{\mathcal{D}}
\newcommand{\I}{\mathcal{I}}
\newcommand{\J}{\mathcal{J}}
\newcommand{\M}{\mathds{M}}
\renewcommand{\P}{\mathds{P}}
\newcommand{\EE}{\mathbf{E}}
\newcommand{\HH}{\mathbf{H}}
\newcommand{\LL}{\mathbf{L}}
\newcommand{\PP}{\mathbf{P}}
\newcommand{\T}{\mathrm{T}}
\newcommand{\id}{\mathbf{I}}
\newcommand{\vl}{\mathrm{V2L}}
\newcommand{\lv}{\mathrm{L2V}}
\newcommand{\cq}{\mathrm{cq}}
\newcommand{\e}{\mathrm{e}}
\newcommand{\tr}{\mathrm{tr}}
\newcommand{\spn}{\mathrm{span}}
\newcommand{\BigO}{\mathcal{O}}
\newcommand{\SmlO}{\mathit{o}}
\newcommand{\ntl}{\mathrm{U}\,}

\usepackage{soul}
\begin{document}

\title[Checking CSL against Quantum CTMCs]{Checking Continuous Stochastic Logic
against Quantum Continuous-Time Markov Chains}

\author[Xu]{Ming Xu\lmcsorcid{0000-0002-9906-5677}}[a]	%required
\address{Shanghai Key Lab of Trustworthy Computing,
East China Normal University, China \\
\& Shanghai Qi Zhi Institute, China}	%required
\email{mxu@cs.ecnu.edu.cn}  %optional
\thanks{XU is supported by National Natural Science Foundation of China (Grant Nos. 12271172 \& 11871221),
and the ``Digital Silk Road'' Shanghai International Joint Lab of Trustworthy Intelligent Software
(Grant No. 22510750100).}	%optional

\author[Mei]{Jingyi Mei\lmcsorcid{0000-0002-4665-9818}}[b]	%optional
\address{Shanghai Key Lab of Trustworthy Computing, East China Normal University, China \\
\& Leiden Institute of Advanced Computer Science, Leiden University, The Netherlands}	%required
\email{mjyecnu@163.com}  %optional
%\thanks{thanks 2, optional.}	%optional

\author[Guan]{Ji Guan\lmcsorcid{0000-0002-3490-0029}}[c]	%optional
\address{Key Laboratory of System Software (Chinese Academy of Sciences)
\& State Key Lab of Computer Science, Institute of Software, Chinese Academy of Sciences, China}	%optional
\email{guanj@ios.ac.cn}
%\urladdr{name3@url3\quad\rm{(optionally, a web-page can be specified)}}  %optional
\thanks{GUAN is supported by the Innovation Program for Quantum Science and Technology (Grant No. 2024ZD0300500),
Youth Innovation Promotion Association, Chinese Academy of Sciences (Grant No. 2023116),
the National Natural Science Foundation of China (Grant No. 62402485),
and the Young Elite Scientists Sponsorship Program, China Association for Science and Technology.}	%optional

\author[Deng]{Yuxin Deng\lmcsorcid{0000-0003-0753-418X}}[d]	%optional
\address{MoE Key Laboratory of Interdisciplinary Research of Computation and Economics, Shanghai University of Finance and Economics, China}	%required
\email{yxdeng@msg.sufe.edu.cn}  %optional
\thanks{DENG is supported by National Natural Science Foundation of China (Grant No. 62472175),
and the Shanghai Trusted Industry Internet Software Collaborative Innovation Center.}

\author[Yu]{Nengkun Yu\lmcsorcid{0000-0003-1188-3032}}[e]	%optional
\address{Department of Computer Science, Stony Brook University, United States}	%optional
\email{nengkunyu@gmail.com}

%% etc.

%% required for running head on odd and even pages, use suitable
%% abbreviations in case of long titles and many authors:

%%%%%%%%%%%%%%%%%%%%%%%%%%%%%%%%%%%%%%%%%%%%%%%%%%%%%%%%%%%%%%%%%%%%%%%%%%%

%% the abstract has to PRECEDE the command \maketitle:
%% be sure not to issue the \maketitle command twice!

\begin{abstract}
  \noindent Verifying quantum systems has attracted a lot of interest in the last decades.
  In this paper,
  we study the quantitative model-checking of quantum continuous-time Markov chains (quantum CTMCs).
  The branching-time properties of quantum CTMCs are specified
  by continuous stochastic logic (CSL),
  which is well-known for verifying real-time systems, including classical CTMCs.
  The core of checking the CSL formulas lies in tackling multiphase until formulas.
  We develop an algebraic method using proper projection, matrix exponentiation, and definite integration
  to symbolically calculate the probability measures of path formulas.
  Thus the decidability of CSL is established.
  To be efficient, numerical methods are incorporated to
  guarantee that the time complexity is polynomial in the encoding size of the input model
  and linear in the size of the input formula.
  A running example of Apollonian networks is further provided to demonstrate our method.  
\end{abstract}

\maketitle

%% \linenumbers
%% start the paper here:
\section{Introduction}\label{S:one}
Rapid development has taken place on quantum hardware in recent years~\cite{ZWD+20,CDG21}.
In the meantime,
quantum software will be crucial in implementing quantum algorithms
and harnessing the power of quantum computers,
e.g. Shor's algorithm with an exponential speed-up for integer factorization~\cite{Sho94}
and Grover's algorithm with a quadratic speed-up for unstructured search~\cite{Gro96}.
The first practical quantum programming language QCL appeared in \"{O}mer's work~\cite{Omer98}.
The quantum guarded command language (qGCL)
was presented to program a ``universal'' quantum computer~\cite{SZ99}.
Selinger~\cite{Sel04} proposed functional programming languages QFC with high-level features.
To ensure reliability,
all kinds of verification techniques~\cite{DHP06,YiF21} need
to be developed for various quantum algorithms and protocols.

Model-checking~\cite{CGP99,BaK08} is one of the most successful technologies
for verifying classical hardware and software systems.
There are also many explorations on applying model-checking to discrete-time quantum systems~\cite{YiF21},
such as quantum programs~\cite{XFJ+24} and quantum (discrete-time) Markov chains~\cite{FYY13}.
However, it is also worth studying continuous-time models in quantum computing.
In this paper,
we consider a novel model of quantum continuous-time Markov chain (quantum CTMC).
It is a composite system consisting of a classical subsystem and a quantum subsystem.
In a classical CTMC,
the instantaneous descriptions (IDs) are given by probability distributions over
a finite set $S$ of classical states.
Whereas, in a quantum CTMC,
the IDs are given by positive semi-definite matrices,
which bear quantum information,
distributed over those classical states $S$.
The ID's evolution of a particular class of quantum CTMCs
that have only one classical state is a \emph{closed} system
in which there are no other classical states (its environment) to interact with.
In comparison, the ID's evolution of a quantum CTMC with more than one classical state is an \emph{open} system.
Let us imagine a swinging pendulum in a mechanical clock.
In an ideal situation, the pendulum does not suffer any unwanted interaction,
such as air friction, from its environment, which is a closed system.
In practice, however, air friction indeed exists, which would eventually make the pendulum cease.
As interactions happen between the pendulum and its environment,
it is an open system.
It is evident that open systems are more common in the real world,
which motivates us to study the open system --- quantum CTMC.
The dynamical system of quantum CTMC can be characterized by the Lindblad master equation:
\begin{equation}\label{Lind}
	\frac{d \rho(t)}{d t}=\L(\rho(t))
\end{equation}
where $\rho(t)$ is the ID of the system at time $t$,
and $\L$ is a linear function of $\rho(t)$ (to be described in Subsection~\ref{S22}).

Under the model of quantum CTMC,
we can develop the notion of cylinder set that is a well-formed set of paths
with a computable probability measure,
which is obtained by
proper \textbf{projection} on the ID $\rho$ in Eq.~\eqref{Lind} for ruling out dissatisfying paths
and \textbf{matrix exponentiation} of the linear function $\L$
for computing the probability distribution as time goes by.
Thereby we establish the probability space to do the formal quantitative reasoning.
As a kind of real-time system,
the branching-time properties are popularly specified by
continuous stochastic logic (CSL).
The core of checking the CSL formulas lies in tackling the multiphase until formula.
It requires the temporal property that
the whole time line should be split into multiple phases,
and in each phase the state of the system should meet the corresponding state constraint,
where the time-stamps switching phases are called switch time-stamps.
More specifically,
the switch time-stamps are in the absolute timing that are measured from the start.
It guides us to design a nested \textbf{definite integral}
concerning the time variables interpreted as the switch time-stamps
to collect the probability measure of the candidate paths phase by phase.
By the mathematical tools of projection,
matrix exponentiation and definite integration,
we can decide the CSL formula
with time complexity being polynomial in the encoding size of the input quantum CTMC
and linear in the size of the input CSL formula,
besides the complexity of the query --- how sufficiently the Euler constant $\e$ should be approached.
However, if measuring the model size by the number of qubits involved,
the decision procedure is still exponential time in that number.
For the consideration of practical efficiency,
we also incorporate numerical methods
--- the scaling and squaring method for matrix exponentiation and the Riemann sum for definite integration ---
to yield a polynomial-time procedure for the satisfaction probability
in the encoding size of the input quantum CTMC.
Furthermore, we provide a running example of Apollonian networks
to demonstrate our method,
and show the usefulness of our method
via checking the percolation performance of those networks.

\subsection{Related Work}
\paragraph{Checking on CTMC}
The IDs for a continuous-time Markov chain (CTMC) are represented by probability distributions.
Aziz \textit{et~al.}~\cite{ASS+96} initiated the model-checking on CTMC in 1996.
They introduced the continuous stochastic logic (CSL) to specify temporal properties of the CTMC.
Roughly speaking, the syntax of CSL amounts to
that of computation tree logic (CTL) plus the multiphase until formula
$\Phi_0 \ntl^{\I_0} \Phi_1 \ntl^{\I_1} \Phi_2 \cdots \ntl^{\I_{K-1}} \Phi_K$
and the probability quantifier formula $\Pr_{>\texttt{c}}(\,\cdot\,)$ defined on those IDs.
They proved the decidability of CSL by number-theoretic analysis.
An approximate model-checking algorithm for a reduced version of CSL
was provided by Baier \textit{et~al.}~\cite{BKH99},
in which multiphase until formulas are restricted to binary until formulas $\Phi_1 \ntl^\I \Phi_2$.
Under this logic, they applied the efficient numerical technique
--- \emph{uniformisation} ---
for \emph{transient analysis}~\cite{BHH+00}.
The approximate algorithms have been extended for multiphase until formulas
using \emph{stratification}~\cite{ZJN+11,ZJN+12}.
After that,
Xu~\textit{et~al.} considered the multiphase until formulas over the CTMC
equipped with a cost/reward structure~\cite{XZJ+16}.
An algebraic algorithm was proposed to attack this problem,
whose effectiveness is ensured by number-theoretic results and algebraic manipulation.
Most of the above algorithms have been implemented in probabilistic model checkers,
like \textsl{PRISM}~\cite{KNP11},
\textsl{Storm}~\cite{DJK+17} and \textsl{EPMC}~\cite{FHL+22}.
However, the IDs of the quantum CTMC
are represented by positive semi-definite matrices,
which generalize the diagonal matrices encoded for probability distributions.
It leads to the fact that quantum CTMCs generalize classical CTMCs,
so new formalism must be developed for this situation,
as done in the current paper.

\paragraph{Checking on quantum DTMC}
A quantum discrete-time Markov chain (quantum DTMC)
is a composite model on classical state space (a finite set)
and quantum state space (a continuum),
on which the evolution is discrete-time by quantum operations~\cite{LiF15}.
Gay \textit{et~al.}~\cite{GNP08}
restricted the quantum operations to Clifford group gates
(i.e., Hadamard, Pauli, CNOT and the phase gates)
and the whole state space as a finite set of describable states, named stabilizer states,
that are closed under those Clifford group gates.
They applied the model checker \textsc{PRISM} to verify the protocols of quantum superdense coding,
quantum teleportation,
and quantum error correction.
Whereas,
Feng \textit{et~al.} proposed the Markov chain with transitions given by the general quantum operations~\cite{FYY13}.
Under the model, the authors considered the reachability probability~\cite{YFY+13},
the repeated reachability probability~\cite{FHT+17},
and the model-checking of a quantum analogy of CTL~\cite{FYY13}.
A key step in their work is
decomposing the whole state space (known as a Hilbert space)
into a direct-sum of some bottom strongly connected component (BSCC) subspaces
plus a maximal transient subspace
concerning a given quantum operation describing the evolution of the quantum DTMC~\cite{YFY+13,GFY18}.
After the decomposition,
all the above problems were shown
to be computable/decidable in polynomial time concerning the dimension of the Hilbert space.
In contrast to the quantum operations in discrete-time evolution,
when considering quantum CTMC,
the general linear operators should be admitted
to describe the rates of state transitions in continuous time.
For this purpose,
we will characterize the continuous-time evolution by the Lindblad master equation in the present work.

\paragraph{Linear-time logic vs branching-time logic}
All of the above literature are concerned with branching-time logic
that classifies all paths into satisfying ones and dissatisfying ones,
and compares the probability measure of satisfying paths with a predefined threshold.
Fewer works concerning linear-time logics
impose constraints on the probability measure of all paths at some critical time periods.
Guan and Yu introduced continuous linear logic (CLL) over CTMCs~\cite{GuY22},
and proved the decidability by real root isolation of exponential polynomials.
Xu \textit{et~al.} initiated the model-checking on quantum CTMCs,
in which the decidability of signal temporal logic (STL) was established, again,
by real root isolation~\cite{XMG+21}.
A novel sample-driven procedure for solving the repeated reachability problem was developed in~\cite{JFX+24}.
Both CLL and STL are linear-time logics.
However, checking the more popular branching-time logic over quantum CTMCs,
as far as we know,
remains open.

\paragraph{Contribution}
The contributions of the present paper can be summarised as:
\begin{enumerate}
	\item We consider the novel model of quantum CTMC that has a more precise description of a composite system
	--- consisting of a classical subsystem and a quantum subsystem --- than the existing work~\cite{XMG+21}.
	\item We prove the decidability of CSL over quantum CTMCs
	and give a numerical approach in time polynomial in the dimension of the state space
	to compute the probability measure.
	\item A running example of Apollonian networks is provided
	to show the utility of our method.
\end{enumerate}

\paragraph{Organization}
The rest of this paper is organized as follows.
Section~\ref{S2} recalls some notions and notations from quantum computing and number theory.
In Section~\ref{S3} we introduce the model of quantum CTMC,
and establish the probability space.
Based on that,
we will check the CSL formula for specifying branching-time properties over quantum CTMCs in Section~\ref{S4}.
We further prove the decidability of CSL and give an efficient numerical approach.
Section~\ref{S5} is the conclusion.
The implementation is delivered in Appendix~\ref{A1}.

\section{Preliminaries}\label{S2}
\subsection{Quantum Computing}
Here we recall some basic notions and notations from quantum computing~\cite{NiC00},
which will be widely used in this paper.
Let $\h$ be a Hilbert space with finite dimension $d$, that is,
a vector space over complex numbers $\mathbb{C}$ equipped with an inner product.
We employ the Dirac notations as follows:
\begin{itemize}[label=$\triangleright$]
	\item $\ket{\psi}$ is a unit column vector labelled with $\psi$;
	\item $\bra{\psi}\coloneqq \ket{\psi}^\dag$ is the Hermitian adjoint
	(transpose and entry-wise complex conjugate) of $\ket{\psi}$;
	\item $\ip{\psi_1}{\psi_2}\coloneqq\bra{\psi_1}\ket{\psi_2}$
	is the inner product of $\ket{\psi_1}$ and $\ket{\psi_2}$;
	\item $\op{\psi_1}{\psi_2}\coloneqq\ket{\psi_1} \otimes \bra{\psi_2}$ is the outer product
	where $\otimes$ denotes tensor product;
	\item $\ket{\psi_1,\psi_2}\coloneqq\ket{\psi_1}\ket{\psi_2}$ is a shorthand of
	the tensor product $\ket{\psi_1}\otimes\ket{\psi_2}$.
\end{itemize}
Fixing an orthonormal basis $\{\ket{\psi_i}\colon 1 \le i \le d\}$ of $\h$,
$\ket{\psi} \in \h$ can be linearly expressed as $\sum_{i=1}^d c_i\ket{\psi_i}$
with $c_i \in \mathbb{C}$ and $\sum_{i=1}^d |c_i|^2=1$.
This $\ket{\psi}$ is called a \emph{superposition}
if it has two or more nonzero terms $c_i\ket{\psi_i}$ under that basis.
For example,
the familiar elements $\ket{\pm}\coloneqq(\ket{0}\pm\ket{1})/\sqrt{2}$ are superpositions
over the basis $\{\ket{0},\ket{1}\}$.

\paragraph{Linear operator}
We will mainly consider the linear operators on $\h$.
Here such a parameter $\h$ can be omitted when it is clear from the context.
A linear operator $\gamma$ is \emph{Hermitian} if $\gamma=\gamma^\dag$;
it is \emph{positive}
if $\bra{\psi}\gamma\ket{\psi} \ge 0$ holds for all $\ket{\psi}\in\h$.
A \emph{projector} $\PP$
is a positive operator of the form $\sum_{i=1}^k \op{\psi_i}{\psi_i}$ ($k\le d$)
for some orthonormal elements $\ket{\psi_i} \in \h$ ($1 \le i \le k$).
Clearly, there is a bijective map
between projectors $\PP=\sum_{i=1}^k \op{\psi_i}{\psi_i}$
and subspaces of $\h$ that are spanned by $\{\ket{\psi_i}\colon 1 \le i \le k\}$.
In summary, positive operators are Hermitian operators
whose eigenvalues are all nonnegative;
projectors are positive operators, all of whose nonzero eigenvalues are $1$.
The \emph{identity operator} $\id$ is the linear operator
$\sum_{i=1}^d\op{\psi_i}{\psi_i}$ for some orthonormal basis $\{\ket{\psi_i}\colon 1 \le i \le d\}$ of $\h$.
A linear operator $\mathbf{U}$ is \emph{unitary}
if $\mathbf{U}\mathbf{U}^\dag=\mathbf{U}^\dag\mathbf{U}=\id$.
The \emph{trace} of a linear operator $\gamma$ is defined as
$\tr(\gamma)\coloneqq\sum_{i=1}^d \bra{\psi_i}\gamma\ket{\psi_i}$
for some orthonormal basis $\{\ket{\psi_i}\colon 1 \le i \le d\}$ of $\h$.
It is worth noting that
the trace function is actually independent of the orthonormal basis selected.
For a linear operator $\gamma_{A,B}$ on the composite Hilbert space $\h_A\otimes\h_B$, 
we can define the \emph{partial traces} as:
\begin{equation}\label{eq:trace}
\begin{aligned}
\tr_A(\gamma_{A,B})
& \coloneqq \sum_{i=1}^{d_A} (\bra{\psi_i^A}\otimes\id_B)\gamma_{A,B}(\ket{\psi_i^A}\otimes\id_B) \\
\tr_B(\gamma_{A,B})
& \coloneqq \sum_{i=1}^{d_B} (\id_A\otimes\bra{\psi_i^B})\gamma_{A,B}(\id_A\otimes\ket{\psi_i^B}),
\end{aligned}
\end{equation}
where $\id_A$ (resp.~$\id_B$) is the identity operator on $\h_A$ (resp.~$\h_B$),
$d_A$ (resp.~$d_B$) is the dimension of $\h_A$ (resp.~$\h_B$),
and $\{\ket{\psi_i^A}\colon 1 \le i \le d_A\}$ (resp.~$\{\ket{\psi_i^B}\colon 1 \le i \le d_B\}$)
is some orthonormal basis of $\h_A$ (resp.~$\h_B$).

\paragraph{Quantum state}
The quantum states are given by
the form of probabilistic ensembles $\{(p_i,\ket{\psi_i})\colon 1\le i \le k\}$
with $p_i>0$ and $\sum_{i=1}^k p_i=1$.
An alternative representation is using positive operators $\rho=\sum_{i=1}^k p_i \op{\psi_i}{\psi_i}$,
named \emph{density operators},
whose traces are unit as $\tr(\rho)=\sum_{i=1}^k p_i \tr(\op{\psi_i}{\psi_i})=1$.
Here those $\ket{\psi_i}$ are not necessarily orthonormal.
To be more explicit,
we resort to the spectral decomposition~\cite[Box~2.2]{NiC00}
that $\rho=\sum_{i=1}^d \lambda_i \op{\lambda_i}{\lambda_i}$,
where $\ket{\lambda_i}$ are eigenvectors interpreted as the \emph{eigenstates} of $\rho$
and $\lambda_i$ are eigenvalues
interpreted as the \emph{probabilities} of taking the eigenstates $\ket{\lambda_i}$.
It is worth noting that the spectral decomposition of $\rho$ is not unique,
but the number of nonzero eigenvalues counted with multiplicities is unique
since it is exactly the rank of $\rho$.
When there is only one eigenstate with positive probability,
$\rho$ is said to be a \emph{pure} state;
otherwise it is a \emph{mixed} state.
In other words, a pure state $\op{\psi}{\psi}$ (or simply $\ket{\psi}$)
indicates the system state which we completely know;
a mixed state $\sum_{i=1}^d \lambda_i \op{\lambda_i}{\lambda_i}$
gives all possible system states $\ket{\lambda_i}$ with a total probability of $\sum_{i=1}^d \lambda_i=1$.
We denote by $\D$ the set of density operators,
and by $\D^{\le 1}$ the set of partial density operators
that are positive operators with trace bounded by $1$.
Sometimes, we need to extract local information from a composite system.
It can be achieved by the partial traces mentioned in Eq.~\eqref{eq:trace},
which result in the \emph{reduced} density operators
$\rho_A=\tr_B(\rho_{A,B})$ on $\h_A$ and $\rho_B=\tr_A(\rho_{A,B})$ on $\h_B$, respectively.
For example,
considering the Bell state $\rho_{A,B} = \tfrac{1}{2}(\ket{0,0}+\ket{1,1})(\bra{0,0}+\bra{1,1})$, 
we can get the reduced density operators
$\rho_A=\tr_B(\rho_{A,B})=\tfrac{1}{2}(\op{0}{0}+\op{1}{1})$
and $\rho_B=\tr_A(\rho_{A,B})=\tfrac{1}{2}(\op{0}{0}+\op{1}{1})$.
In other words,
the mixed states $\rho_A$ on $\h_A$ and $\rho_B$ on $\h_B$ are traced
from the pure state $\rho_{A,B}$ on $\h_A \otimes \h_B$.

\paragraph{Quantum operation}
The quantum operations (also known as super-operators) are \emph{completely positive} operators that
map from the set $\D$ of density operators to itself.
As a special case, all unitary transformations $\{\mathbf{U}\}$ are quantum operations
that act on pure states as $\ket{\psi} \mapsto \mathbf{U}\ket{\psi}$
and on mixed states as $\rho \mapsto \mathbf{U}\rho\mathbf{U}^\dag$.
Generally speaking, every quantum operation can be written
in Kraus representation $\{\EE_j:1 \le j \le m\}$ for some $m \in \mathbb{N}$,
where $\EE_j$ are Kraus operators satisfying $\sum_{j=1}^m \EE_j^\dag\EE_j = \id$,
so that it acts on quantum states as $\rho \mapsto \sum_{j=1}^m \EE_j \rho \EE_j^\dag$.
Here, the condition $\sum_{j=1}^m \EE_j^\dag\EE_j = \id$ is the so-called \emph{trace-preservation},
due to
\begin{equation}
	\tr\left(\sum_{j=1}^m \EE_j \rho \EE_j^\dag\right)
	=\sum_{j=1}^m \tr(\EE_j \rho \EE_j^\dag)
	=\sum_{j=1}^m \tr(\EE_j^\dag\EE_j \rho)
	=\tr\left(\sum_{j=1}^m \EE_j^\dag\EE_j \rho \right)
	=\tr(\rho).
\end{equation}
It is worth noting that the Kraus representation of a quantum operation is not unique,
but the number $m$ of Kraus operators can be bounded by $d^2$,
since quantum operations are linear operators on $\D$
and $\D$ is contained in the set of Hermitian operators which is a linear space of dimension $d^2$.

\paragraph{State evolution}
A \emph{closed} system is
an ideal system that does not suffer from any unwanted interaction from outside environment.
State evolution in a closed system is induced by
some unitary transformation $\ket{\psi(t+\Delta t)}=\mathbf{U}(\Delta t)\ket{\psi(t)}$,
where $\ket{\psi(t)}$ is the state of the system at time $t$.
From the identity $\mathbf{U}(\Delta t)=\exp(-\imath \HH \Delta t)$,
we can determine the \emph{Hamiltonian} (a Hermitian operator) $\HH$
in the Schr\"odinger equation characterizing that state evolution:
\begin{equation}\label{eq:schrodinger}
	\frac{d \ket{\psi(t)}}{d t} = -\imath \HH\ket{\psi(t)}.
\end{equation}
An \emph{open} system interacts with its environment.
Composed with the environment, the large system is closed.
In more detail,
for a mixed state $\rho=\sum_{i=1}^k p_i \op{\psi_i}{\psi_i}$ in an open system,
there is a pure state in the large system
\[
	\hat{\rho}=\left(\sum_{i=1}^k \sqrt{p_i}\ket{\psi_i,\textup{env}_i}\right)
	\left(\sum_{i=1}^k \sqrt{p_i}\bra{\psi_i,\textup{env}_i}\right)
\]
for some orthonormal states $\ket{\textup{env}_i}$ of the environment.
After the purification,
state evolution in the large system is induced by the unitary transformation
$\hat{\rho}(t+\Delta t) = \hat{\mathbf{U}}(\Delta t) \hat{\rho}(t) \hat{\mathbf{U}}^\dag(\Delta t)$ as above.
By tracing out the environment,
we get the quantum operation in the original system:
\[
	\rho(t+\Delta t)
	=\sum_{i=1}^k (\id\otimes\bra{\textup{env}_i}) \hat{\rho}(t+\Delta t) (\id\otimes\ket{\textup{env}_i})
	=\sum_j \EE_j(\Delta t) \rho(t) \EE_j^\dag(\Delta t),
\]
where the last equation comes from Kraus representation of the current quantum operation.
Based on those Kraus operators $\EE_j$ that are additionally assumed to be Markovian,
we will obtain the Lindblad master equation
characterizing the state evolution of the open system
in the coming subsection.

\subsection{Lindblad Master Equation}\label{S22}
The quantum continuous-time Markov chain considered in this paper is
a kind of open systems induced by the quantum operation
\begin{equation}\label{eq1}
	\rho(t+\Delta t)=\sum_{j=0}^m \EE_j \rho(t) \EE_j^\dag.
\end{equation}
The Lindblad master equation~\cite{Lin76,GKS76} can be obtained
by knowing the structure of those operators $\EE_j$.
The dynamical system like other Markov models has the \emph{memoryless} property,
i.e., its state evolution is determined only by the current state.
It requires that $\EE_j$ should depend only on the infinitesimal time $\Delta t$,
not on the time $t$.
So it suffices to expand $\EE_j$ into power series,
particularly determining the coefficients of $(\Delta t)^0=1$, $(\Delta t)^{1/2}=\sqrt{\Delta t}$,
$(\Delta t)^1=\Delta t$, and so on if necessary.

First, by $\lim_{\Delta t \to 0}\rho(t+\Delta t)=\rho(t)$,
we can assume without loss of generality that
i) there exists an operator, say $\EE_0(\Delta t)$,
that amounts to $\id \cdot 1 + \SmlO(1)$,
where the infinitesimal $\SmlO((\Delta t)^k)$ means
$\lim_{\Delta t \to 0}\SmlO((\Delta t)^k)/(\Delta t)^k=0$ for $k\ge 0$,
and ii) other operators $\EE_j(\Delta t)$ ($j>0$) amount to $\SmlO(1)$.

Next, supposing $\EE_j(\Delta t)=\LL_j \sqrt{\Delta t}+ \SmlO(\sqrt{\Delta t})$ ($j>0$),
we get
\begin{equation}\label{eq2}
	\sum_{j=1}^m \EE_j(\Delta t) \rho(t) \EE_j^\dag(\Delta t)
	= \sum_{j=1}^m \LL_j \rho(t) \LL_j^\dag \Delta t + \SmlO(\Delta t).
\end{equation}
But the quantum operation $\{\LL_j: 1 \le j \le m\}$ is not trace-preserving.
To resolve it, $\EE_0(\Delta t)$ is further supposed to be
$\id -\imath\HH\Delta t -\tfrac{1}{2}\sum_{j=1}^m\LL_j^\dag\LL_j\Delta t +\SmlO(\Delta t)$
where $\HH$ is Hermitian,
so that
\begin{equation}\label{eq3}
	\EE_0(\Delta t) \rho(t) \EE_0^\dag(\Delta t)
	= \rho(t) +\left[-\imath\HH\rho(t) + \imath\rho(t)\HH
	-\tfrac{1}{2}\sum_{j=1}^m (\LL_j^\dag\LL_j\rho(t) + \rho(t)\LL_j^\dag\LL_j)\right] \Delta t + \SmlO(\Delta t),
\end{equation}
while the quantum operation
$\{\id -\imath\HH\Delta t -\tfrac{1}{2}\sum_{j=1}^m\LL_j^\dag\LL_j\Delta t\} \cup \{\LL_j: 1 \le j \le m\}$
is trace-preserving
and therefore $\sum_{j=0}^m \EE_j^\dag(\Delta t)\EE_j(\Delta t)=\id + \SmlO(\Delta t)$.

Finally, putting Eqs.~\eqref{eq1}, \eqref{eq2} and~\eqref{eq3} into
\[
	\rho'=\lim_{\Delta t \to 0} \dfrac{\rho(t+\Delta t)-\rho(t)}{\Delta t},
\]
we obtain the Lindblad master equation
for characterizing the quantum continuous-time Markov chain.
It is of the following form
\begin{equation}\label{eq:Lindblad}
\begin{aligned}
	\rho' & =-\imath\HH\rho+\imath\rho\HH
	+\sum_{j=1}^m \left(\LL_j\rho \LL_j^\dag
	-\tfrac{1}{2}\LL_j^\dag\LL_j\rho-\tfrac{1}{2}\rho\LL_j^\dag\LL_j\right) \\
	& =-\imath[\HH,\rho]+
	\sum_{j=1}^m\left(\LL_j\rho \LL_j^\dag-\tfrac{1}{2}\{\LL_j^\dag\LL_j,\rho\}\right),
\end{aligned}
\end{equation}
where $[\mathbf{A},\mathbf{B}]\coloneqq\mathbf{A}\mathbf{B}-\mathbf{B}\mathbf{A}$
denotes the \emph{commutator} between two linear operators $\mathbf{A}$ and $\mathbf{B}$,
and $\{\mathbf{A},\mathbf{B}\}\coloneqq\mathbf{A}\mathbf{B}+\mathbf{B}\mathbf{A}$
denotes the \emph{anti-commutator}.
The term $-\imath[\HH,\rho]$ describes the evolution of the internal system;
the term $\LL_j\rho \LL_j^\dag-\tfrac{1}{2}\{\LL_j^\dag\LL_j,\rho\}$ 
represents the interaction between the system and its environment.
In other words,
to characterize the evolution of an open system,
it is necessary to use those linear operators $\LL_j$ besides the Hermitian one $\HH$.

For the given quantum operation~\eqref{eq1},
those linear operators $\LL_j$ ($j>0$) are determined by $\EE_j$
and the Hermitian operator $\HH$ is determined by $\EE_0$.
Conversely, once we are given the operators $\LL_j$ ($j>0$) and $\HH$,
the quantum operation~\eqref{eq1} are also determined.
It entails that the Lindblad master equation is
the most general type of Markovian and time-homogeneous master equation
describing (generally non-unitary) state evolution that preserves the laws of quantum mechanics,
i.e., complete positivity and trace-preservation.
Interested readers can refer to~\cite[Section~3.5]{Pre15}
for more physical explanations about the state evolution.

We now turn to derive the solution of Eq.~\eqref{eq:Lindblad}.
Two useful functions are defined as:
\begin{itemize}[label=$\triangleright$]
	\item $\lv(\gamma)\coloneqq\sum_{i=1}^d \sum_{j=1}^d \bra{i}\gamma\ket{j} \ket{i,j}$
	that rearranges entries of the linear operator $\gamma$
	on the Hilbert space $\h$ with dimension $d$
	as a column vector; and
	\item $\vl(\mathbf{v})\coloneqq\sum_{i=1}^d \sum_{j=1}^d \bra{i,j} \mathbf{v} \op{i}{j}$
	that rearranges entries of the column vector $\mathbf{v}$ as a linear operator.
\end{itemize}
Here, $\lv$ and $\vl$ are pronounced
``linear operator to vector'' and ``vector to linear operator'', respectively.
They are mutually inverse functions,
so that
if a linear operator (resp.~its vectorization) is determined,
its vectorization (resp.~the original linear operator) is determined.
Hence, we can freely choose one of the two representations for convenience.
For any linear operators $\mathbf{A},\mathbf{B},\mathbf{C}$,
their product $\mathbf{D}=\mathbf{A}\mathbf{B}\mathbf{C}$ has the transformation
\[
	\lv(\mathbf{D})=(\mathbf{A} \otimes \mathbf{C}^\T)\lv(\mathbf{B}),
\]
where $\T$ denotes transpose.
Then, we can reformulate Eq.~\eqref{eq:Lindblad} as the linear ordinary differential equation
\begin{equation}\label{eq:ODE}
	\lv(\rho') = \M \cdot\lv(\rho),
\end{equation}
where $\M=-\imath\HH\otimes\id+\imath\id\otimes\HH^\T
+\sum_{j=1}^m \big( \LL_j\otimes\LL_j^{*}
-\tfrac{1}{2}\LL_j^\dag\LL_j\otimes\id
-\tfrac{1}{2}\id\otimes \LL_j^\T\LL_j^* \big)$ is the \emph{governing matrix}
with $*$ denoting entry-wise complex conjugate.
As a result,
we get the desired solution $\lv(\rho(t))=\exp(\M\cdot t)\cdot\lv(\rho(0))$
or equivalently $\rho(t)=\vl(\exp(\M\cdot t)\cdot\lv(\rho(0)))$ in a closed form.
It is obtained in polynomial time by the standard method~\cite{Kai80}.

\subsection{Number Theory}
We will show the decidability of a temporal logic,
which is based on the following essential facts in number theory.
\begin{defi}
	A number $\alpha$ is \emph{algebraic},
	denoted $\alpha \in \mathbb{A}$,
	if there is a nonzero $\mathbb{Q}$-polynomial $f_\alpha(z)$,
	satisfying $f_\alpha(\alpha)=0$;
	otherwise $\alpha$ is \emph{transcendental}.
\end{defi}
\noindent In the above definition, such a polynomial $f_\alpha(z)$ is called
the \emph{minimal polynomial} of $\alpha$ if it is irreducible.
The \emph{degree} $D$ of $\alpha$ is $\deg_z(f_\alpha)$.
A popular encoding of $\alpha$~\cite[Subsection~4.2.1]{Coh93}
is using the minimal polynomial $f_\alpha$
plus an isolation disk in the complex plane
that distinguishes $\alpha$ from other roots of $f_\alpha$.
So the encoding size $\|\alpha\|$ is the number of bits used
to store the minimal polynomial and the isolation disk.

\begin{defi}
	Let $\mu_1,\ldots,\mu_m$ be irrational complex numbers.
	Then the \emph{field extension} $\mathbb{Q}(\mu_1,\ldots,\mu_m):\mathbb{Q}$
    is the smallest set
	that contains $\mu_1,\ldots,\mu_m$ and is closed under arithmetic operations,
	i.e. addition, subtraction, multiplication and division.
\end{defi}
\noindent Here those irrational complex numbers $\mu_1,\ldots,\mu_m$ are called
the generators of the field extension.
A field extension is \emph{simple} if it has only one generator.
For instance, the field extension $\mathbb{Q}(\sqrt{2}):\mathbb{Q}$
is exactly the set $\{a+b\sqrt{2} \colon a,b \in \mathbb{Q}\}$.

\begin{lemC}[{\cite[Algorithm~2]{Loo83}}]\label{lem:simple}
	Let $\alpha_1$ and $\alpha_2$ be two algebraic numbers of
	degrees $D_1$ and $D_2$, respectively.
	There is an algebraic number $\mu$ of degree at most $D_1 D_2$,
	such that the field extension $\mathbb{Q}(\mu):\mathbb{Q}$
	is exactly $\mathbb{Q}(\alpha_1,\alpha_2):\mathbb{Q}$.
\end{lemC}
\noindent For a collection of algebraic numbers $\alpha_1,\ldots,\alpha_m$
appeared in the input instance,
by repeatedly applying Lemma~\ref{lem:simple},
we can obtain a simple field extension $\mathbb{Q}(\mu):\mathbb{Q}$
that can generate all $\alpha_1,\ldots,\alpha_m$.

\begin{lemC}[{\cite[Corollary~4.1.5]{Coh93}}]\label{lem:closed}
	Let $\alpha$ be an algebraic number of degree $D$,
	and $g(z)$ a polynomial with degree $D_g$
	and coefficients taken from $\mathbb{Q}(\alpha):\mathbb{Q}$.
	There is a $\mathbb{Q}$-polynomial $f(z)$ of degree at most $DD_g$,
	such that the roots of $g(z)$ are those of $f(z)$.
\end{lemC}
\noindent The above lemma entails that
roots of any polynomial with coefficients taken from algebraic numbers
($\mathbb{A}$-polynomial for short)
are also algebraic.
	
\begin{thmC}[(Lindemann 1882) {\cite[Theorem~1.4]{Bak75}}]\label{Lindemann}
	For any nonzero algebraic numbers $\alpha_1,\ldots,\alpha_m$ and
	any distinct algebraic numbers $\lambda_1,\ldots,\lambda_m$,
	the sum $\sum_{i=1}^m \alpha_i \mathrm{e}^{\lambda_i}$ with $m \ge 1$ is nonzero.
\end{thmC}
\noindent Hence the sign of a given real number
of the form $\sum_{i=1}^m \alpha_i \mathrm{e}^{\lambda_i}$
is decidable,
which can be achieved by sufficiently approximating the Euler constant $\e$.
But the complexity of the procedure is unknown as studied in the existing literature~\cite{COW16,HLX+18}.
It becomes a query operation (also known as oracle) in our work.

\section{Quantum Continuous-Time Markov Chain}\label{S3}
In this section, we first introduce the model of quantum continuous-time Markov chain,
which subsumes the classical continuous-time Markov chain.
Based on that, we will establish a probability space over paths
to reason about real-time properties.
An algebraic approach using matrix exponentiation is presented
to calculate those measurable events in the probability space.

\begin{defi}
	Let $AP$ be a set of atomic propositions throughout this paper.
	A labelled \emph{quantum continuous-time Markov chain} (quantum CTMC for short) $\QC$
	over the Hilbert space $\h$ is a triple $(S,Q,L)$, in which
	\begin{itemize}
		\item $S$ is a finite set of classical states;
		\item $Q$ is the transition generator function given by two parts of information:
		\begin{itemize}
			\item a Hermitian operator $\HH$
			of the form $\sum_{s\in S} \op{s}{s} \otimes \HH_s$
			where $\HH_s$ are Hermitian operators on $\h$,
			\item a finite set of linear operators $\LL_j$ of the form
			$\op{s_{j,2}}{s_{j,1}} \otimes \LL_{s_{j,1},s_{j,2}}$
			for some $s_{j,1},s_{j,2}\in S$, ($s_{j,1} \ne s_{j,2}$),
			where $\LL_{s_{j,1},s_{j,2}}$ are linear operators on $\h$;
			and
		\end{itemize}
		\item $L\colon S \to 2^{AP}$ is a labelling function
		that labels each classical state with some atomic propositions in $AP$.
	\end{itemize}
\end{defi}

To distinguish the dynamic states $\rho(t)$ of $\QC$ in time $t$
from other static (classical or quantum) states,
we call $\rho(t)$ by instantaneous descriptions (IDs).
Usually, a density operator of the form $\rho(0)=\sum_{s\in S} \op{s}{s} \otimes \rho_s$
is appointed as the initial ID of $\QC$,
where $\rho_s \in \D_\h^{\le 1}$ satisfy $\sum_{s\in S} \tr(\rho_s)=1$.

Every classical state $s \in S$ can be quantized by a pure state $\ket{s}$,
satisfying that all of them are orthogonal pairwise.
Specifically, if the system is in a classical state $s$,
it can be thought of being in the quantum state $\ket{s}$ with unit probability mass.
Further, the superpositions, say $c_1\ket{s_1}+c_2\ket{s_2}$,
can be quantized states between the basis states $\ket{s_1}$ and $\ket{s_2}$
with probability mass $|c_1|^2$ and $|c_2|^2$ respectively.
Let $\mathcal{C}\coloneqq\spn(\{\ket{s}\colon s\in S\})$ be the Hilbert space corresponding to the classical system,
and $\h_\cq\coloneqq\mathcal{C} \otimes \h$ the enlarged Hilbert space
corresponding to the whole classical--quantum system.
The dimension of $\h_\cq$ is $N\coloneqq nd$ where $n=|S|$ and $d=\dim(\h)$.
This parameter $N$ is used to reflect the encoding size of the quantum CTMC $\QC$.

The IDs $\rho$ for a quantum CTMC $\QC$ are represented by density operators on the enlarged Hilbert space $\h_\cq$
with the mixed structure $\sum_{s\in S} \op{s}{s} \otimes \rho_s$
where $\rho_s \in \D_\h^{\le 1}$ are partial density operators,
satisfying $\sum_{s\in S} \tr(\rho_s)=1$.
Such a mixed structure entails that
the entries $(\bra{s_1} \otimes \id) \rho (\ket{s_2} \otimes \id)$
are zero whenever $s_1 \ne s_2$.
It underlies the key difference
between the classical state space $\mathcal{C}$ and the quantum state space $\h$ that:
\begin{itemize}
	\item The classical state space $\mathcal{C}$ has finitely many basis states $\ket{s}$ with $s\in S$,
	which are known and fixed a priori.
	After tracing out the quantum state space $\h$,
	the reduced density operator on $\mathcal{C}$ is a probability distribution over these classical states,
	and cannot be a \emph{superposition}.
	\item The quantum state system $\h$ has uncountably many pure states $\ket{\psi}$,
	due to the postulate of quantum mechanics~\cite{NiC00}
	that the space of pure states is a \emph{continuum}.
	Those pure states $\ket{\psi}$ are superpositions on some orthonormal basis of $\h$.
	If the orthonormal basis is chosen properly,
	the pure states $\ket{\psi}$ would be expressed much explicitly, thus facilitating the subsequential computation.
	But properly choosing in the orthonormal basis cannot be guaranteed in advance.
	It will bring extra technical hardness to verifying quantum models.
\end{itemize}
We can regard these IDs $\rho$ as positive-operator valued distributions over classical states $S$,
which generalize probability distributions representing the IDs of a classical CTMC.
Under the mixed structure,
the current model of quantum CTMC is more precise than
the one proposed in~\cite{XMG+21} that does \emph{not} clarify
the difference between the classical state space and the quantum state space.

The transition generator function $Q$
is functionally similar to the transition rate matrix
in the classical CTMC,
but employs the Lindblad master equation to characterize the continuous-time transition.
It gives rise to a \emph{universal} way to describe the behavior of a quantum CTMC
that keeps the ID $\rho$ in the mixed structure $\sum_{s\in S} \op{s}{s} \otimes \rho_s$,
following from the generality of the Lindblad master equation.
In more detail, we can see:
\begin{itemize}
\item Each $\HH_s$ determines the internal evolution at classical state $s$,
i.e., $-\imath [\HH_s,\rho_s]$ is one term of $\rho_s'$,
and $\HH=\sum_{s\in S} \op{s}{s} \otimes \HH_s$ gives all such individual internal evolution
by $-\imath [\HH,\rho]$.
\item Whereas, $\LL_j=\op{s_{j,2}}{s_{j,1}}\otimes\LL_{s_{j,1},s_{j,2}}$
determines the external evolution from classical state $s_1$ to $s_2$,
i.e., $\LL_j\rho\LL_j^\dag-\tfrac{1}{2}\{\LL_j,\rho\}$ is one term appeared in $\rho'$.
\end{itemize}
It is notable that, as required in Eq.~\eqref{eq:Lindblad},
$\HH_s$ can be chosen as \emph{arbitrary} Hermitian operators
and $\LL_{s_{j,1},s_{j,2}}$ can be \emph{arbitrary} linear operators,
which makes the model very expressive.
However, there is also an assumption $s_{j,1} \ne s_{j,2}$
on linear operators $\LL_j = \op{s_{j,2}}{s_{j,1}} \otimes \LL_{s_{j,1},s_{j,2}}$,
so the model is less general than possible.
The assumption follows the convention of classical CTMC that
the system has no interactions with the outside environment except when changing classical states.
For another, the Hilbert space $\h_s$ for each classical state $s$ can be chosen to be different,
but they are same here only for facilitating the subsequential analysis.
A state $s$ is \emph{absorbing} if all transitions from it are removed,
which is achieved by some projectors to be defined later.

Recall that a classical CTMC $\mathfrak{C}$ is a triple $(S,\mathbf{Q},L)$,
where $S$ and $L$ are the same as the components in the quantum CTMC $\QC=(S,Q,L)$,
and $\mathbf{Q}\colon (S \times S) \to \mathbb{R}_{\ge 0}$ is the transition rate matrix,
i.e., $\mathbf{Q}[s_1,s_2]$ gives the transition rate from state $s_1$ to $s_2$.
In fact, the model of quantum CTMC $\QC=(S,Q,L)$ extends
that of classical CTMC $\mathfrak{C}=(S,\mathbf{Q},L)$,
so the model of quantum CTMC is more expressive.
It can be seen from the following lemma:
\begin{lemC}[{\cite[Lemma~10]{XMG+21}}]
	Given a CTMC $\mathfrak{C}$,
	it can be modelled by a quantum CTMC $\QC$ over one-dimensional Hilbert space $\h$.
\end{lemC}

For the sake of computability,
the entries of $\rho_s$, $\HH_s$ and $\LL_{s_{j,1},s_{j,2}}$
in the model of quantum CTMC
are supposed to be algebraic numbers
in the simple field extension $\mathbb{Q}(\mu):\mathbb{Q}$ for some algebraic number $\mu$.
The encoding size $\|\mu\|$ is used to reflect the encoding size of operands in $\QC$.
Together with the dimension $N$ of $\h_\cq$,
they are two key parameters in the encoding size $\|\QC\|$,
which will be used to analyze the complexity of our method afterward.

\begin{exa}\label{ex1}
	Here we will study the open quantum walk (OQW) along Apollonian networks (ANs)~\cite{AHA+05}.
	ANs have small-world and scale-free properties.
	The small-world property means that
	there is a relatively short path between two nodes despite the often large size of networks,
	and the scale-free property implies that
	most of the nodes in the network are connected to only a few nodes 
	while few nodes are connected to a large number of nodes.
	With the two properties, 
	ANs provide an excellent facility to analyze dynamical processes taking place on networked systems, 
	including percolation and electrical conduction.
    In the meanwhile, ANs serve as a valuable tool for studying properties of various walk models,
    including (discrete- and continuous-time) random walks and quantum walks~\cite{XLL08}.
    
	ANs are generated as follows:
	At the initial generation, 
	the network consists of three nodes marked as 0, 1 and 2, respectively.
	At each subsequent generation, 
	a new node is added inside each triangle and linked to the three vertices of that triangle,
	as shown in~\figurename~\ref{fig:apollonian}.

\begin{figure}
	\begin{subfigure}[b]{0.32\textwidth}
	\centering
	\resizebox{\linewidth}{!}{
	\begin{tikzpicture}[->,>=stealth',auto,node distance=2.5cm,semithick,inner sep=2pt,
	state/.style={circle,draw,minimum size=0.5cm}]
		\node[state,draw] at (2,{2*sqrt(3)}) (s0){$0$};
		\node[state,draw] at  (0,0) (s1){$1$};
		\node[state,draw] at  (4,0) (s2){$2$};
		\draw (s0)  edge[-]node[below]{} (s1);
		\draw (s1)  edge[-]node[above,rotate=-60]{} (s2);
		\draw (s2)  edge[-]node[above,rotate=60]{} (s0);	
	\end{tikzpicture}
	}
	\caption{0th generation}\label{fig:apollonian0}
	\end{subfigure}
	\begin{subfigure}[b]{0.32\textwidth}
	\centering
	\resizebox{\linewidth}{!}{
	\begin{tikzpicture}[->,>=stealth',auto,node distance=2.5cm,semithick,inner sep=1pt,
	state/.style={circle,draw,minimum size=0.5cm}]
		\node[state,draw] at (2,{2*sqrt(3)}) (s0){$0$};
		\node[state,draw] at  (0,0) (s1){$1$};
		\node[state,draw] at  (4,0) (s2){$2$};
		\node[state,draw] at  (2,{2/sqrt(3)}) (s3){$3$};	
		\draw (s0)  edge[-]node[below]{}node[font=\scriptsize,above,rotate=60]{$\leftarrow A$}
		node[font=\scriptsize,below,rotate=60]{$B \rightarrow$} (s1);
		\draw (s1)  edge[-]node[font=\scriptsize,above]{$\leftarrow B$}
		node[font=\scriptsize,below]{$A \rightarrow$} (s2);
		\draw (s2)  edge[-]node[font=\scriptsize,above,rotate=-60]{$\leftarrow A$}
		node[font=\scriptsize,below,rotate=-60]{$B \rightarrow$} (s0);
		\draw (s2)  edge[-]node[font=\tiny,above,rotate=-30,xshift=-4]{$A+C/\sqrt{3}\rightarrow$}
		node[font=\scriptsize,below,rotate=-30]{$\leftarrow C$} (s3);
		\draw (s3)  edge[-]node[font=\scriptsize,above,rotate=90]{$\leftarrow C$}
		node[font=\scriptsize,below,rotate=90]{$\tfrac{1}{\sqrt{3}}C \rightarrow$} (s0);
		\draw (s3)  edge[-]node[font=\tiny,above,rotate=30,xshift=4]{$\leftarrow B+C/\sqrt{3}$}
		node[font=\scriptsize,below,rotate=30]{$C \rightarrow$} (s1);
	\end{tikzpicture}
	}
	\caption{1st generation}\label{fig:apollonian1}
	\end{subfigure}
	\begin{subfigure}[b]{0.32\textwidth}
	\centering
	\resizebox{\linewidth}{!}{
	\begin{tikzpicture}[->,>=stealth',auto,node distance=2.5cm,semithick,inner sep=1pt,
	state/.style={circle,draw,minimum size=0.5cm}]
		\node[state,draw] at (2,{2*sqrt(3)}) (s0){$0$};
		\node[state,draw] at  (0,0) (s1){$1$};
		\node[state,draw] at  (4,0) (s2){$2$};
		\node[state,draw] at  (2,{2/sqrt(3)}) (s3){$3$};
		\node[state,draw,minimum size=0.2cm] at  ({2*sqrt(3)-2},{2*sqrt(3)-2}) (s4){$4$};
		\node[state,draw,minimum size=0.2cm] at  (2,{4-2*sqrt(3)}) (s5){$5$};
		\node[state,draw,minimum size=0.2cm] at  ({6-2*sqrt(3)},{2*sqrt(3)-2}) (s6){$6$};				
		\draw (s0)  edge[-]node[below]{} (s1);
		\draw (s1)  edge[-]node[above,rotate=-60]{} (s2);
		\draw (s2)  edge[-]node[above,rotate=60]{} (s0);
		\draw (s2)  edge[-]node[above,rotate=90]{} (s3);
		\draw (s3)  edge[-]node[below,rotate=40]{} (s0);
		\draw (s3)  edge[-]node[below,rotate=-40]{} (s1);
		\draw (s4)  edge[-]node[above,rotate=90]{} (s0);
		\draw (s4)  edge[-]node[below,rotate=40]{} (s1);
		\draw (s4)  edge[-]node[below,rotate=-40]{} (s3);
		\draw (s5)  edge[-]node[above,rotate=90]{} (s1);
		\draw (s5)  edge[-]node[below,rotate=40]{} (s2);
		\draw (s5)  edge[-]node[below,rotate=-40]{} (s3);
		\draw (s6)  edge[-]node[above,rotate=90]{} (s0);
		\draw (s6)  edge[-]node[below,rotate=40]{} (s2);
		\draw (s6)  edge[-]node[below,rotate=-40]{} (s3);
	\end{tikzpicture}
	}
	\caption{2nd generation}\label{fig:apollonian2}
    \end{subfigure}
	\caption{The first three generations of ANs}\label{fig:apollonian}
\end{figure}
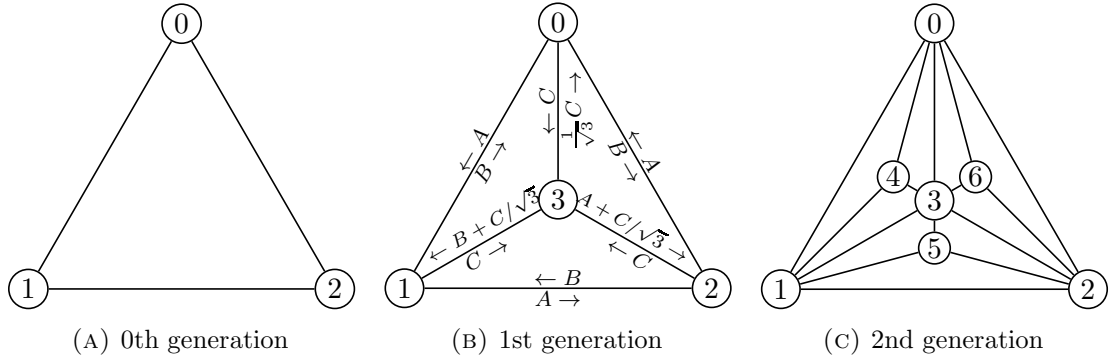

    Consider a qutrit, a quantum state in a 3-dimensional Hilbert space,
	is walking along the first generation of AN~\cite{PGM+15},
	which is depicted in \figurename~\ref{fig:apollonian1}.
	As the qutrit may interact with its environment,
	some noise from the network channel will affect the walking of the qutrit
	which can be described by linear operators,
	say, $A=\op{x}{x}$, $B=\op{y}{y}$ and $C=\op{z}{z}$,
	where $\ket{x}=(1,1,1)^\T/\sqrt{3}$, $\ket{y}=(1,w,w^2)^\T/\sqrt{3}$ and $\ket{z}=(1,w^2,w)^\T/\sqrt{3}$
	with $w=\exp(2\pi\imath/3)$.
	The OQW can be modelled by the quantum CTMC $\QC_1=(S,Q,L)$
	over the Hilbert space $\h=\spn(\{\ket{0},\ket{1},\ket{2}\})$,
	where $S=\{0,1,2,3\}$ represents the four nodes;
	Node $3$ has a label ``\texttt{center}'' and other nodes have none;
	the transition function $Q$ is given by the Hermitian operator $\HH=0$
	and the following 12 linear operators:
	\[
	\begin{aligned}
		& \LL_1 = \op{1}{0}\otimes A
		&& \LL_2 = \op{2}{0}\otimes B 
		&& \LL_3 = \op{3}{0}\otimes C \\
		& \LL_4 = \op{2}{1}\otimes A	
		&& \LL_5 = \op{0}{1}\otimes B
		&& \LL_6 = \op{3}{1}\otimes C \\
		& \LL_7 = \op{0}{2}\otimes A
		&& \LL_8 = \op{1}{2}\otimes B	
		&& \LL_9 = \op{3}{2}\otimes C \\
		& \LL_{10} = \op{0}{3}\otimes C/\sqrt{3} 
		&& \LL_{11} = \op{1}{3}\otimes (B+C/\sqrt{3})
		&& \LL_{12} = \op{2}{3}\otimes (A+C/\sqrt{3}).
	\end{aligned}
	\]
	An initial ID $\rho(0)$ could be, say, $\op{3}{3}\otimes\tfrac{1}{3}\id$.
	Here, the left operand $\op{3}{3}$ in the tensor product means that
	we are definitely in Node $3$ in the classical system,
	the right operand $\tfrac{1}{3}\id$ means that
	we are in one of three basis states $\ket{0}$ through $\ket{2}$
	with equal probability $\tfrac{1}{3}$ in the quantum system,
	and they together determine the ID in the whole classical--quantum system. \qed
\end{exa}

A path quantified over the classical states of a quantum CTMC $\QC$ is an infinite sequence
$\omega = s_0 \xrightarrow{\tau_0} s_1 \xrightarrow{\tau_1} s_2 \xrightarrow{\tau_2} \cdots$.
(For conciseness we call those paths quantified over the classical states simply by paths afterwards.)
It means that $\omega$ starts at $s_0$, and sojourns there up to time $t_0=\tau_0$ (but excluding $\tau_0$);
then $\omega$ jumps to $s_1$ at time $t_0$,
and sojourns there up to time $t_1=\tau_0+\tau_1$;
then $\omega$ jumps to $s_2$ at time $t_1$,
and sojourns there up to time $t_2=t_1+\tau_2$;
and so on.
The times $\tau_k$ ($k \ge 0$) are called \emph{sojourn} times
that are measured from the instants
when the prior transitions $s_{k-1} \rightarrow s_k$ take place;
while the times $t_k$ ($k \ge 0$) are \emph{switch} time-stamps
that are measured from the start of the path.
A finite path $\hat{\omega}$ is a finite fragment of such a path $\omega$,
e.g. $s_0 \xrightarrow{\tau_0} s_1 \xrightarrow{\tau_1} s_2 \xrightarrow{\tau_2} \cdots
\xrightarrow{\tau_{K-1}} s_K$.
Namely, the state $\omega(t)$ of $\omega$ at time $t$ is $s_k$
for $t\in[t_{k-1},t_k)$,
where $t_{-1}$ is defined to be zero for convenience.
Let $Path$ be the set of paths of $\QC$.

To formally reason about quantitative properties of quantum CTMC,
we will establish the probability space over paths as follows:
\begin{defi}
	A \emph{measurable space} is a pair $(\Omega,\Sigma)$,
	where $\Omega$ is a nonempty set
	and $\Sigma$ is a $\sigma$-algebra on $\Omega$ that is a collection of subsets of $\Omega$,
	satisfying:
	\begin{itemize}
		\item $\Omega \in \Sigma$, and
		\item $\Sigma$ is closed under countable union and complement.
	\end{itemize}
	In addition, a \emph{probability space} is a triple $(\Omega,\Sigma,\Pr)$,
	where $(\Omega,\Sigma)$ is a measurable space
	and $\Pr\colon \Sigma \to [0,1]$ is a probability measure,
	satisfying:
	\begin{itemize}
		\item $\Pr(\Omega)=1$,
		\item $\Pr(A) \ge 0$ for any $A \in \Sigma$, and
		\item $\Pr(\biguplus_i A_i) = \sum_i \Pr(A_i)$ for any pairwise disjoint $A_i \in \Sigma$.
	\end{itemize}
\end{defi}
Note that any finite path $\hat{\omega}=
s_0 \xrightarrow{\tau_0} s_1 \xrightarrow{\tau_1} s_2 \xrightarrow{\tau_2} \cdots
\xrightarrow{\tau_{K-1}} s_K$ is of probability measure zero,
since it requires that the transitions from $s_{k-1}$ to $s_k$ ($k>0$) take place instantaneously.
It cannot collect any positive probability from a countable union of such events.
To fix it, we use the \emph{cylinder set} of paths,
which is defined as
\begin{equation}\label{eq:cylinder}
	\begin{aligned}
	& Cyl(s_0 \xrightarrow{\J_0} s_1 \xrightarrow{\J_1} s_2 \xrightarrow{\J_2} \cdots
	\xrightarrow{\J_{K-1}} s_K) \coloneqq \\
	& \qquad \left\{\omega \in Path\colon \omega=s_0 \xrightarrow{\tau_0} s_1 \xrightarrow{\tau_1} s_2
	\xrightarrow{\tau_2} \cdots \xrightarrow{\tau_{K-1}} s_K \xrightarrow{\tau_K} \cdots
	\wedge \bigwedge_{k=0}^{K-1}\tau_k\in\J_k\right\}.
	\end{aligned}
\end{equation}
Here, these time intervals $\J_k$ ($0 \le k < K$) are \emph{relative},
as they are measured from the instants
when the prior transitions $s_{k-1} \rightarrow s_k$ take place.
Let $\Omega=Path$,
and $\Pi \subseteq 2^\Omega$ be countably many cylinder sets $Cyl$ of paths
plus the empty set $\emptyset$.
By~\cite[Chapter~10]{BaK08},
there is a smallest $\sigma$-algebra $\Sigma$ of $\Pi$ containing $\Pi$,
such that the pair $(\Omega,\Sigma)$ forms a measurable space.

Next, for a cylinder set
$Cyl=s_0 \xrightarrow{\J_0} s_1 \xrightarrow{\J_1} s_2 \xrightarrow{\J_2}
\cdots \xrightarrow{\J_{K-1}} s_K$,
we define the probability measure along $Cyl$
as $\Pr(Cyl)=\tr(\rho^{(K)})$,
where $\rho^{(K)}$ is the partial density operator in the $K$th phase.
It can be calculated inductively by
\begin{equation}\label{eq:inductive}
	\left\{
	\begin{aligned}
		\rho^{(0)} =&\ \PP_{s_0} \cdot \rho(0) \\
		\rho^{(k)} =&\ \PP_{s_k} \cdot \vl(\exp(\M_{s_{k-1}} |\J_{k-1}|) \ \cdot \\
		&\ \lv(\PP_{s_{k-1}} \cdot \vl(\exp(\M_{s_{k-1}} \inf \J_{k-1}) \cdot \lv(\rho^{(k-1)}))))
		& \mbox{if } k>0, \\
	\end{aligned}
	\right.
\end{equation}
where $|\J_{k-1}|=\sup\J_{k-1}-\inf\J_{k-1}$ is the length of $\J_{k-1}$,
and $\M_s$ is the governing matrix adapted for
the transition generator function $Q_s$ consisting of
the Hermitian operator $\HH\cdot\PP_s$ and finitely many linear operators $\LL_j\PP_s$.
It is preferable to calculate those matrix multiplications in Eq.~\eqref{eq:inductive}
with right associativity, according to the forward semantics.
Specifically,
$\PP_{s_{k-1}} \cdot \rho$ extracts
the term $\op{s_{k-1}}{s_{k-1}} \otimes \rho_{s_{k-1}}$
from the ID $\rho=\sum_{s\in S} \op{s}{s} \otimes \rho_{s}$;
$\exp(\M_{s_{k-1}} \tau)$ gives rise to the state transition function,
which carries the density operator of state $s_{k-1}$
to all states $s$ ($s\in S$) for some sojourn time $\tau$
and preserves the density operator of state $s$ ($s \ne s_{k-1}$).
Totally, the induction is
computing the final partial density operator $\rho^{(k)}$ from $\rho^{(k-1)}$,
satisfying a sojourn at $s_{k-1}$ for some time $\tau_{k-1}\in\J_{k-1}$.
Similar to Vardi's work for classical Markov chains~\cite{Var85},
the domain of $\Pr$ can be extended to $\Sigma$,
i.e., $\Pr\colon\Sigma\to [0,1]$,
which is well-defined under the countable union $\bigcup_i A_i$ for any $A_i \in \Sigma$
and the complement $A^\textup{c}$ for any $A \in \Sigma$.
Hence the triple $(\Omega,\Sigma,\Pr)$ forms a probability space.

\begin{exa}\label{ex2}
	Let us continue to consider
	Example~\ref{ex1}.
	The cylinder set $Cyl_1=3\xrightarrow{(0,1)} 1 \xrightarrow{(1,2)} 3$
	contains all paths that take the transitions from Node $3$ to Node $1$ and then back to Node $3$
	with sojourn time in the intervals $\J_1 = (0,1)$ and $\J_2 = (1,2)$ respectively.
	We will compute its probability with the initial ID $\rho(0)=\op{3}{3}\otimes\op{0}{0}$.
	Firstly, we have the projectors $\PP_i = \op{i}{i} \otimes \id$ for $i=0,1,2,3$
	and the corresponding adapted governing matrices
	\[
		\M_i = \sum_{j=1}^{12} [ (\LL_j\PP_i)\otimes(\LL_j\PP_i)^*
		-\tfrac{1}{2}(\LL_j\PP_i)^\dag(\LL_j\PP_i)\otimes\id
		-\tfrac{1}{2}\id\otimes (\LL_j\PP_i)^\T(\LL_j\PP_i)^* ].
	\]
	Then, we inductively calculate the partial density operator in the $k$th phase:
	\[
	\begin{aligned}
		\rho^{(0)} & = \PP_{3} \cdot \rho(0) 
		= \op{3}{3}\otimes\op{0}{0}, \\
		\rho^{(1)} & = \PP_{1} \cdot \vl(\exp(\M_{3} |\J_1|)
		\cdot \lv(\PP_{3} \cdot \vl(\exp(\M_{3} \inf \J_1) \cdot
		\lv(\rho^{(0)})))) = \op{1}{1}\otimes \rho_1, \\
		\rho^{(2)} & = \PP_{3} \cdot \vl(\exp(\M_{1} |\J_2|)
		\cdot \lv(\PP_{1} \cdot \vl(\exp(\M_{1} \inf \J_2) \cdot
		\lv(\rho^{(1)})))) = \op{3}{3}\otimes \rho_3,
	\end{aligned}
	\]
	in which $\rho_1$ and $\rho_3$ are computable partial density operators on $\h$
	with the following numerical approximations obtained
	by the scaling and squaring method embedded with Pad\'{e} approximations~\cite{MVL03,AMH09}
	for matrix exponentials.
	That is, $\exp(\M)=\big(\exp(\M/m)\big)^m$
	and $\exp(\M/m) \approx \mathds{P}_k$ for some $k,m$,
	where $\mathds{P}_k$ is the $k$th Pad\'{e} approximation to be described in Eq.~\eqref{eq:Pade},
	\[
	\begin{aligned}
		\rho_1 \approx \ & 
		0.184318\op{0}{0} + (-0.092159 - 0.037415 \imath)[\op{0}{1} + \op{2}{0}] \ + \\
		& (-0.092159 + 0.037415 \imath)[\op{0}{2} + \op{1}{0}] + 0.053744 \op{1}{1} \ + \\
		& (0.038416 - 0.037415\imath)\op{1}{2} + (0.038416 + 0.037415 \imath) \op{2}{1} + 0.053744\op{2}{2}], \\
		\rho_3 \approx \ & 
		0.016333\op{0}{0} + (-0.008166 + 0.014145\imath)[\op{0}{1} + \op{1}{2}] \ + \\
		& (-0.008166 - 0.014145\imath)[\op{0}{2} + \op{1}{0}] + 0.016333\op{1}{1} \ + \\
		& (-0.008166 + 0.014145\imath)\op{2}{0} + (-0.008166 - 0.014145\imath) \op{2}{1} + 0.016333\op{2}{2}.
	\end{aligned}
	\]
	Finally we get the probability $\Pr(Cyl_1)=\tr(\rho^{(2)})=\tr(\rho_3) \approx 0.048999$. \qed
\end{exa}

\paragraph{Complexity}
By Eq.~\eqref{eq:inductive},
the procedure of exactly calculating the probability $\Pr(Cyl)$ requires $K$ induction steps,
for each of which the matrix exponentiation $\exp(\M)$ dominates the computational cost.
The matrix exponential $\exp(\M)$ can be computed by Jordan decomposition as follows.
Since $\M$ is a matrix of dimension $N^2$
with entries being algebraic numbers taken from the simple field extension $\mathbb{Q}(\mu):\mathbb{Q}$,
the characteristic polynomial $g(x)$ of $\M$ is an $\mathbb{A}$-polynomial with degree $N^2$,
which can be obtained in $\BigO((N^2)^4)=\BigO(N^8)$ by~\cite[Algorithm~8.17]{BPR06}.
To further determine the eigenvalues of $\M$,
we first convert $g(x)$ to a $\mathbb{Q}$-polynomial $f(x)$.
It is achieved by Sylvester resultant~\cite[Notation~4.12]{BPR06} ---
that is a determinant of dimension $2\deg(\mu)$ ---
to eliminate the common occurrence of $\mu$ in $g(x)$ and $f_\mu$ (the minimal polynomial of $\mu$).
It results in the desired $\mathbb{Q}$-polynomial $f(x)$
with degree at most $D\coloneqq N^2 \deg(\mu)$
has all roots of $g(x)$.
The variable elimination is in $\BigO(D^3)$ by the standard determinant computation.
The roots of $f(x)$, including the eigenvalues of $\M$,
can be determined in $\BigO(D^5 \log(D)^2)$ by~\cite[Algorithm~10.4]{BPR06}.
So we get the Jordan decomposition $\M=\mathds{T} \cdot \mathds{J} \cdot \mathds{T}^{-1}$
in $\BigO((N^2)^3)=\BigO(N^6)$
by computing the eigenvectors with respect to the known eigenvalues,
where $\mathds{J}$ is the Jordan canonical form of $\M$
and $\mathds{T}$ is the transformation matrix.
The same complexity holds for
computing $\exp(\M \tau)=\mathds{T} \cdot \exp(\mathds{J} \tau) \cdot \mathds{T}^{-1}$.
Hence the exactly calculating procedure is in
\begin{equation}
	\BigO(K \cdot D^5 \log(D)^2)
\end{equation}
that is bounded polynomially in the encoding size $\|\QC\|$
and linearly in the length $K=|Cyl|$.
However, if we measure the size of $\QC$ by the number $\log d$ of qubits
used to make up the Hilbert space $\h$ of dimension $d$,
the complexity would turn out to be exponential in $\log d$ as $D= N^2 \deg(\mu)$ and $N=nd$.
It implies that the method cannot scale the large-qubit quantum state system well.

The above method aims to calculate the probability $\Pr(Cyl)$ exactly,
so that the decidability result could be established on later,
i.e. deciding $\Pr(Cyl)\sim \textup{c}$ for some appointed probability threshold $c \in [0,1]$,
where $\mathord{\sim} \in \{\mathord{=}, \mathord{<}, \mathord{>}\}$ is a comparison operator.
However, if we focus on the numerical value of $\Pr(Cyl)$,
an approximating method, say the following one, suffices,
which is potentially more efficient.
Here, we have to address the problem that
the uniformisation method~\cite{BHH+00} is not directly applicable to compute $\exp(\M)$,
since the governing matrix $\M$ has no structure of the generator matrix in a classical CTMC
that is discretized as a Poisson process.

\paragraph{Numerical speed-up}
To improve efficiency,
we employ the state-of-the-art scaling and squaring method~\cite{MVL03,AMH09}
for the general matrix exponentiation.
The technique is \emph{reliable},
as it gives some warnings whenever it introduces excessive errors.
The rationale and steps are described as follows.
\begin{enumerate}
\item We first \emph{scale} $\exp(\M)$ by choosing $m$ to be a power of two,
satisfying $\|\exp(\M/m)\| \le 1$ where $\|\,\cdot\,\|$ denotes Frobenius norm.
Obviously, such a number $m$ exists,
which can further be chosen to be linear in the size $N^4$ of $\M$,
i.e. $m \in \BigO(N^4)$.
\item As a subroutine, we proceed to approximate $\exp(\M/m)$ reliably.
The Pad\'{e} approximation provides such a subroutine if $\|\exp(\M/m)\|$ is small,
which is what have been done in the first step.
The $k$th Pad\'{e} approximation $\mathds{P}_k$ is the fraction
\begin{equation}\label{eq:Pade}
	\left(\sum\limits_{i=0}^k \dfrac{(2k-i)!\cdot k! \cdot (\M/m)^i}{(2k)!\cdot (k-i)!\cdot i!}\right)
	\Big/
	\left(\sum\limits_{i=0}^k \dfrac{(2k-i)!\cdot k!\cdot (-\M/m)^i}{(2k)!\cdot (k-i)!\cdot i!}\right),
\end{equation}
and its computational cost lies in $k$ times of matrix multiplications,
for each of which is in $\BigO((N^2)^3)=\BigO(N^6)$.
The subroutine terminates
whenever the difference $\|\mathds{P}_{k+1}-\mathds{P}_k\|$
of two successive approximations $\mathds{P}_k$ and $\mathds{P}_{k+1}$
is less than a predefined tolerance $\epsilon$.
As the sequence $\{ \mathds{P}_k \}_{k \in \mathbb{N}}$ is exponentially convergent to $\exp(\M/m)$,
i.e., the error is exponentially convergent to $0$,
we get the desired index $k \in \BigO(\log(1/\epsilon))$
or equivalently there is a $k \in \BigO(\|\epsilon\|)$
that yields the approximation $\mathds{P}_k$ within error $\epsilon$ of $\exp(\M/m)$.
\item Finally we keep \emph{squaring} $\exp(\M/m)$ by computing the powers $\mathds{P}_k^{2^l}$
for $l=1,2,\ldots,\log(m)$ in turn,
and eventually get the approximation $\mathds{P}_k^m$ of $\exp(\M)$,
which costs $\log(m)$ times of matrix multiplications.
\end{enumerate}
Hence, the approximately calculating procedure is in
\begin{equation}
	\BigO(K \cdot (k + \log(m)) \cdot N^6)=\BigO(K \cdot (\|\epsilon\|+\log(N)) \cdot N^6).
\end{equation}
At least a factor $N^4$ in the complexity of the exactly calculating procedure is saved here.

\section{Checking Continuous Stochastic Logic}\label{S4}
We introduce the continuous stochastic logic (CSL for short) that consists of three layers:
i) state formulas (which can be true or false in a specific classical state),
ii) path formulas (which can be true or false along a specific path),
and iii) model formulas (true or false for a specific model).
Then we present an algebraic algorithm for checking the CSL formulas over quantum CTMCs.

\begin{defi}
	The syntax of the CSL formulas,
	consisting of state formulas $\Phi$,
	path formulas $\phi$ and model formulas $\chi$,
	are defined as follows:
	\[
	\begin{aligned}
		\Phi & \coloneqq \textup{a} \mid \neg\Phi \mid \Phi_1 \wedge \Phi_2 \\
		\phi & \coloneqq \Phi_0 \ntl^{\I_0} \Phi_1 \ntl^{\I_1} \Phi_2 \cdots \ntl^{\I_{K-1}} \Phi_K \\
		\chi & \coloneqq \mathcal{P}_{\sim \textup{c}}(\phi)
	\end{aligned}
	\]
	where $\textup{a} \in AP$ is an atomic proposition,
	$\I_k \subseteq \mathbb{R}_{\ge 0}$ ($0 \le k <K$) are time intervals with rational endpoints,
	$\mathord{\sim} \in \{\mathord{=}, \mathord{<}, \mathord{>}\}$ is a comparison operator,
	and $\textup{c} \in [0,1] \cap \mathbb{Q}$ is a probability threshold.
\end{defi}
In particular, the model formula $\chi$ is also referred to as the CSL formula.
For convenience, we assume that
those time intervals $\I_k$ are disjoint left-open and right-closed intervals $(a_k,b_k]$,
and $b_{K-1}$ is allowed to be $\infty$.
Otherwise some detailed techniques could be introduced to deal with as in~\cite[Section~5]{ZJN+12}.
The path formula $\phi$ is called the \emph{multiphase until} formula
to be interpreted below.

Following~\cite{ASS+96,BHH+00}, we give the formal semantics of CSL.
\begin{defi}
	The semantics of CSL interpreted over a quantum CTMC $\QC=(S,Q,L)$
	is given by the satisfaction relation $\models$:
	%\begin{align*}
	%	s & \models \textup{a}
	%	&& \textup{if } \textup{a} \in L(s), \\
	%	s & \models \neg\Phi
	%	&& \textup{if } s \not\models \Phi, \\
	%	s & \models \Phi_1 \wedge \Phi_2
	%	&& \textup{if } s \models \Phi_1 \textup{ and } s \models \Phi_2, \\
	%	\omega & \models \Phi_0 \ntl^{\I_0} \Phi_1 \ntl^{\I_1} \Phi_2 \cdots \ntl^{\I_{K-1}} \Phi_K
	%	&& \textup{if there exist }t_0 < t_1 < \cdots < t_{K-1} \textup{ such that for each integer }0 \le k < K, \\
	%	&&& \textup{we have }t_k \in \I_k,
	%	\textup{ and } \\
	%	&&& \textup{for each }t' \in [t_{k-1},t_k)\colon \omega(t') \models \Phi_k
	%	\textup{ where }t_{-1}\textup{ is defined to be zero}, \\
	%	&&& \textup{and additionally }\omega(t_{K-1}) \models \Phi_K, \\
	%	\QC & \models \mathcal{P}_{\sim \textup{c}}(\phi)
	%	&& \textup{if } \Pr(\{\omega \colon \omega \models \phi\}) \sim \textup{c}.
	%\end{align*}

	\renewcommand\arraystretch{1.2}
		\begin{tabular}{llll}
			$s$ & $\models \textup{a}$ && if $\textup{a} \in L(s)$; \\
			$s$ & $\models \neg\Phi$ && if $s \not\models \Phi$; \\
			$s$ & $\models \Phi_1 \wedge \Phi_2$ && if $s \models \Phi_1$ and $s \models \Phi_2$; \\
			$\omega$ & $\models \Phi_0 \ntl^{\I_0} \Phi_1 \ntl^{\I_1} \Phi_2 \cdots \ntl^{\I_{K-1}} \Phi_K$
			&& if there exist $t_0 < t_1 < \cdots < t_{K-1}$ \\
			& \multicolumn{3}{l}{\qquad such that for each integer $0 \le k < K$, we have $t_k \in \I_k$} \\
			& \multicolumn{3}{l}{\qquad and $\forall\,t' \in [t_{k-1},t_k) \,:\, \omega(t') \models \Phi_k$
				where $t_{-1}$ is defined to be zero,} \\
			& \multicolumn{3}{l}{\qquad and additionally $\omega(t_{K-1}) \models \Phi_K$;} \\
			$\QC$ & $\models \mathcal{P}_{\sim \textup{c}}(\phi)$
			&& if $\Pr(\{\omega \,:\, \omega \models \phi\}) \sim \textup{c}$.
		\end{tabular}
\end{defi}
\noindent Here,
the real numbers $t_1,\ldots,t_{K-1}$ in the interpretation of the multiphase until formula $\phi$
are switch time-stamps,
and the time intervals $\I_k$ ($0 \le k < K$) in $\phi$ are \emph{absolute} timing
that are measured from the start of the path,
which differs from the LTL interpretation of the nested until formula with the \emph{relative} timing
that are measured from the instants when the prior transitions take place.

The core of the model-checking algorithm lies in measuring all the paths
satisfying the multiphase until formula $\phi$,
which can be tackled by the following lemma.
Before that, we define the following four notions.
\begin{enumerate}
	\item Let $\PP_\Phi\coloneqq\sum_{s\models\Phi} \op{s}{s} \otimes \id$
	be the projector for some state formula $\Phi$,
	and $\P_\Phi=\PP_\Phi \otimes \PP_\Phi$.
	\item The transition generator function $Q_\Phi$ is given by
	the Hermitian operator $\HH_\Phi\coloneqq\HH\cdot\PP_\Phi$
	and the finite set of linear operators $\LL_{\Phi,j}\coloneqq\LL_j\PP_\Phi$.
	Then, the governing matrix $\M_\Phi$ adapted for $Q_\Phi$ is
	\begin{equation}\label{eq:gov}
		-\imath\HH_\Phi\otimes\id+\imath\id\otimes\HH_\Phi^\T
		+\sum_{j=1}^m \left( \LL_{\Phi,j}\otimes\LL_{\Phi,j}^*
		-\tfrac{1}{2}\LL_{\Phi,j}^\dag\LL_{\Phi,j}\otimes\id
		-\tfrac{1}{2}\id\otimes \LL_{\Phi,j}^\T\LL_{\Phi,j}^* \right).
	\end{equation}
	The transition generator function $Q_\Phi$ and the corresponding governing matrix $\M_\Phi$
	keep the behaviors from the $\Phi$-states and absorb the other states.
	\item Recall in probability theory that
	for an exponentially distributed random variable with rate parameter $\kappa$,
	the \emph{cumulative distribution function} $\mathrm{CDF}(t)$ is $1-\exp(\kappa t)$
	that determines the probability of the random variable bounded from above by $t$.
	So the change in the probability for the random variable falling into a time interval $(a,b]$
	is $\mathrm{CDF}(b)-\mathrm{CDF}(a)=\exp(\kappa a)-\exp(\kappa b)$.
	Correspondingly, the matrix exponential $\exp(\M t)$ is the state transition function
	that determines the change in the density operator representing ID in a quantum CTMC.
	\item Again, in probability theory,
	the \emph{probability density function} $\mathrm{PDF}(t)$ is $\kappa \exp(\kappa t)$
	that determines the change rate in the probability of the random variable,
	since $\mathrm{CDF}(b)-\mathrm{CDF}(a)=\int_a^b \mathrm{PDF}(t) \ d t$.
	Correspondingly,
	$\M \cdot \exp(\M t)$ determines the change rate in the density operator.
\end{enumerate}
They are powerful tools to develop the algorithm for checking the CSL formula over quantum CTMCs,
together with definite integrals as those for classical CTMCs in~\cite{XZJ+16}.

\begin{lem}\label{lem:multiphase}
	For a multiphase until formula
	$\phi=\Phi_0 \ntl^{\I_1} \Phi_1 \ntl^{\I_2} \Phi_2 \cdots \ntl^{\I_K} \Phi_K$,
	the final partial density operator $\rho^{(K)}$ measuring
	all the paths that satisfy $\phi$
	can be calculated inductively by
	\begin{equation}\label{eq:cal}
		\lv(\rho^{(k)}) = \begin{cases}
		    \P_{\Phi_0} \cdot \lv(\rho(0)) \qquad k=0, \\
			\P_{\Phi_k} \cdot [\exp(\M_{\Phi_{k-1}} (b_k-b_{k-1}))
			+ \int_{a_k}^{b_k} \exp(\M_{\Phi_k} (b_k-t_k)) \cdot \P_{\Phi_{k-1}} \ \cdot \\
			\qquad \P_{\neg\Phi_{k-1}} \cdot \M_{\Phi_{k-1}} \cdot \exp(\M_{\Phi_{k-1}} (t_k-a_k)) \ d t_k \ \cdot \\
			\qquad \exp(\M_{\Phi_{k-1}} (a_k-b_{k-1}))] \cdot \lv(\rho^{(k-1)})
			\qquad \mbox{if } 0<k<K, \\
			\P_{\Phi_k} \cdot
			[\exp(\M_{\Phi_{k-1} \wedge \neg\Phi_k}(b_k-a_k)) \cdot \P_{\Phi_{k-1}} + \P_{\Phi_{k-1}}] \ \cdot \\
			\qquad \exp(\M_{\Phi_{k-1}} (a_k-b_{k-1})) \cdot \lv(\rho^{(k-1)})
			\qquad \mbox{if } k=K,
		\end{cases}
	\end{equation}
	where $a_k=\inf \I_k$, $b_k=\sup \I_k$ and $b_0$ is defined to be zero for convenience;
	and thus the satisfaction probability $\Pr(\{\omega \colon \omega \models \phi\})$
	is $\tr(\rho^{(K)})$.
\end{lem}
\begin{proof}
	We prove it by an induction on the phase index $k$.

	Initially, $\rho^{(0)}$ is the partial density operator
	$\sum_{s\models \Phi_0} \op{s}{s} \otimes \rho_s$
	extracted from $\rho(0)=\sum_{s\in S} \op{s}{s} \otimes \rho_s$
	by the projector $\PP_{\Phi_0}$ (see Notion~1),
	which measures the paths that satisfy $\phi$ at time $0$.

	Let $\rho^{(k-1)}$ ($k<K$) be the partial density operator
	measuring the paths that satisfy $\phi$ up to time $b_{k-1}$.
	We proceed to construct the partial density operator $\rho^{(k)}$
	measuring the paths that satisfy $\phi$ up to time $b_k$ as follows.
	Based on the semantics of the $k$th phase of $\phi$,
	all the paths that satisfy $\phi$
	can be classified into two disjoint path sets:
	\begin{itemize}
	\item one contains the paths leaving
	$\Phi_{k-1}$-states for $\Phi_k$-ones during $\I_k$,
	i.e., for some switch time-stamp $t_k \in \I_k=(a_k,b_k]$,
	\[
		\{ \omega \colon
		[\forall\,t'\in(b_{k-1},t_k)\colon \omega(t') \models \Phi_{k-1}]
			\wedge [\omega(t_k) \models \neg\Phi_{k-1}]
			\wedge [\forall\,t''\in[t_k,b_k]\colon \omega(t'') \models \Phi_k] \},
	\]
	\item the other contains the paths that do not have to leave
	$\Phi_{k-1}$-states during $\I_k$,
	i.e.
	$\{ \omega \colon [\forall\,t'\in(b_{k-1},b_k]\colon \omega(t') \models \Phi_{k-1}]
	\wedge [\omega(b_k) \models \Phi_k] \}$.
	\end{itemize}
	Here, ``leaving $\Phi_{k-1}$-states'' means that
	the path $\omega$ satisfies $\phi$ but at some $t \in \I_k$, $\omega(t)$ dissatisfies $\Phi_{k-1}$
	and it must enter into $\Phi_k$-states then.
	Correspondingly, the former path set is measured by the partial density operator
	\[
		\vl\left(\begin{aligned}
			& \P_{\Phi_k} \cdot \int_{a_k}^{b_k} \exp(\M_{\Phi_k} (b_k-t_k)) \cdot
			\P_{\neg\Phi_{k-1}} \cdot \M_{\Phi_{k-1}} \cdot
			\exp(\M_{\Phi_{k-1}} (t_k-a_k)) \ d t_k \ \cdot \\
			& \P_{\Phi_{k-1}} \cdot \exp(\M_{\Phi_{k-1}} (a_k-b_{k-1}))
		\cdot \lv(\rho^{(k-1)})
	\end{aligned}\right),
	\]
	where
	\begin{itemize}
		\item $\M_{\Phi_{k-1}}$ and $\M_{\Phi_k}$ are the governing matrices
		respectively adapted to $Q_{\Phi_{k-1}}$ and $Q_{\Phi_k}$ (see Notion~2),
		\item $\exp(\M_{\Phi_{k-1}} (a_k-b_{k-1}))$ and $\exp(\M_{\Phi_k} (b_k-t_k))$
		are the state transition functions (see Notion~3), and
		\item $\M_{\Phi_{k-1}} \cdot \exp(\M_{\Phi_{k-1}} (t_k-a_k))$ gives the rate of state transition (see Notion~4).
	\end{itemize}
	The latter path set is measured by the partial density operator
	$\vl(\P_{\Phi_k} \cdot \exp(\M_{\Phi_{k-1}} (b_k-b_{k-1})) \cdot \lv(\rho^{(k-1)}))$.
	Since the two path sets are disjoint,
	$\rho^{(k)}$ is the sum of the above two partial density operators.

	Let $\rho^{(K-1)}$ be the partial density operator
	measuring the paths that satisfy $\phi$ up to time $b_{K-1}$.
	We proceed to construct the partial density operator $\rho^{(K)}$
	measuring the paths that satisfy $\phi$ up to time $b_K$ as follows.
	Based on the semantics of the $K$th phase of $\phi$,
	all the paths that satisfy $\phi$
	are classified into two disjoint sets:
	\begin{itemize}
	\item one is not immediately reaching $\Phi_K$-states during $\I_K$,
	i.e., for some switch time-stamp $t_K \in \I_K=(a_K,b_K]$,
	\[
		\{ \omega \colon
			[\forall\,t'\in[b_{K-1},t_K)\colon \omega(t') \models \Phi_{K-1}]
			\wedge [\forall\,t''\in[a_K,t_K)\colon \omega(t'') \models \neg\Phi_K]
			\wedge [\omega(t_K) \models \Phi_K] \},
	\]
	\item the other is immediately reaching $\Phi_K$-states at $a_K$, i.e.
	$\{ \omega \colon
	[\forall\,t'\in[b_{K-1},a_K)\colon \omega(t') \models \Phi_{K-1}]
	\wedge [\omega(a_K) \models \Phi_K] \}$.
	\end{itemize}
	Here, ``not immediately reaching $\Phi_K$-states'' means that
	the satisfying path $\omega$ dissatisfies $\Phi_K$ at $a_K$,
	and it enters in $\Phi_K$-states for some $t>a_K$.
	Correspondingly, the two disjoint path sets are measured by the partial density operators:
	\[
		\begin{aligned}
		& \vl\big(\P_{\Phi_k} \cdot \exp(\M_{\Phi_{k-1} \wedge \neg\Phi_k}(b_k-a_k)) \cdot
		\P_{\Phi_{k-1}} \cdot \exp(\M_{\Phi_{k-1}} (a_k-b_{k-1})) \cdot \lv(\rho^{(k-1)})\big) \\
		& \vl(\P_{\Phi_k \wedge \Phi_{k-1}} \cdot \exp(\M_{\Phi_{k-1}} (a_k-b_{k-1}))
		\cdot \lv(\rho^{(k-1)})),
		\end{aligned}
	\]
	whose sum is exactly $\rho^{(K)}$.
	Hence $\rho^{(K)}$ is the final partial density operator
	measuring the paths that satisfy the whole path formula $\phi$.
\end{proof}

When tackling $b_K=\infty$,
we do not suffer from the trouble that
the bottom strongly connected component (BSCC) subspaces will make
the orbit of density operators periodical,
since the probability $\tr(\rho^{(K)})$ of satisfying $\phi$
is monotonically increasing along $b_K$.
So it can be achieved by taking the limit.

\begin{exa}\label{ex3}
	The percolation performance is an essential property of networks
	that describes transitions between nodes and their neighborhood.
	Reconsidering the first generation of AN in Example~\ref{ex1},
	the return probability of the \texttt{center} node reflects the percolation performance.
	We aim to compute the probability of the following return event.
	\begin{quote}
		\textit{The qutrit starts at the \texttt{center} node,
			then transits to non-center nodes in one unit of time and sojourns there up to (absolute) time at least $1$,
			and finally returns to the \texttt{center} node within time $2$.}
	\end{quote}
	This could be formally specified by the CSL formula
	$\phi_1 \equiv \Phi_0 \ntl^{\I_1}\Phi_1\ntl^{\I_2} \Phi_2$,
	where
	$\Phi_0 \equiv \Phi_2 \equiv \texttt{center}$,
	$\Phi_1 \equiv \neg\texttt{center}$,
	and time intervals $\I_1 = (0,1]$ and $\I_2=(1,2]$.
	Firstly, we take $\rho(0) = \op{3}{3}\otimes\tfrac{1}{3}\id$ and define
	\[
	\begin{aligned}
		\P_{\Phi_0} &= \P_{\Phi_2} = \PP_{\Phi_0}\otimes\PP_{\Phi_0}
		= (\op{3}{3}\otimes\id)\otimes(\op{3}{3}\otimes\id), \\
		\P_{\Phi_1} &= \P_{\neg\Phi_0} =\P_{\Phi_1 \wedge \neg\Phi_2} = \PP_{\Phi_1}\otimes\PP_{\Phi_1} \\
		&= [(\op{0}{0}+\op{1}{1}+\op{2}{2})\otimes\id]\otimes[(\op{0}{0}+\op{1}{1}+\op{2}{2})\otimes\id], \\
		\M_{\Phi_i} &= \sum_{j=1}^{12} [ (\LL_j\PP_{\Phi_i})\otimes(\LL_j\PP_{\Phi_i})^{*}
		-\tfrac{1}{2}(\LL_j\PP_{\Phi_i})^\dag(\LL_j\PP_{\Phi_i})\otimes\id
		-\tfrac{1}{2}\id\otimes (\LL_j\PP_{\Phi_i})^\T(\LL_j\PP_{\Phi_i})^* ],
	\end{aligned}
	\]
	for $i=0$ or $1$.
	Then, we calculate $\rho^{(2)}$ in an inductive fashion:
	\[
	\begin{aligned}
		\lv(\rho^{(0)}) = & \
		\P_{\Phi_0} \cdot \lv(\rho(0)), \\
		\lv(\rho^{(1)}) = & \ \P_{\Phi_1} \cdot [\exp(\M_{\Phi_{0}} (1-0))
		+ \int_{0}^{1} \exp(\M_{\Phi_1} (1-t_1)) \cdot \P_{\neg\Phi_{0}} \cdot \M_{\Phi_{0}} \ \cdot \\
		& \qquad \exp(\M_{\Phi_{0}} (t_1-0)) \ d t_1 \cdot \P_{\Phi_{0}} \cdot \exp(\M_{\Phi_{0}} (0-0))]
		\cdot \lv(\rho^{(0)}), \\
		\lv(\rho^{(2)}) = & \ [\P_{\Phi_2} \cdot \exp(\M_{\Phi_1 \wedge \neg\Phi_2}(2-1))+\P_{\Phi_2}]
		\cdot \P_{\Phi_1} \cdot \exp(\M_{\Phi_1} (1-1)) \cdot \lv(\rho^{(1)}).
	\end{aligned}
	\]
	The result is $\rho^{(2)}= \op{3}{3}\otimes\rho_3$
	with computable partial density operator $\rho_3$ on $\h$.
	We replace the definite integral with the Riemann sum under a predefined step length,
	say $\tfrac{1}{100}$,
	\[
		\frac{1}{100} \sum_{i=1}^{100}
		\exp(\M_{\Phi_1} \tfrac{100-i}{100}) \cdot \P_{\neg\Phi_0} \cdot \M_{\Phi_0}
		\cdot\exp(\M_{\Phi_0} \tfrac{i}{100}),
	\]
	and get the numerical approximation
	\[
	\begin{aligned}	
		\rho_3 \approx \ & 
		0.105266 [\op{0}{0} + \op{1}{1} + \op{2}{2}]
		+ (-0.052633 + 0.091163\imath)[\op{0}{1} + \op{1}{2} + \op{2}{0}] \ + \\
		& (-0.052633 - 0.091163\imath)[\op{0}{2} + \op{1}{0} + \op{2}{1}].
	\end{aligned}
	\]
Thus we obtain the return probability $\Pr(\phi)=\tr(\rho^{(2)})=\tr(\rho_3) \approx 0.315798$. \qed
\end{exa}

\paragraph{Complexity}
The procedure of exactly calculating the final density operator by Eq.~\eqref{eq:cal}
is based on the exact matrix exponentiation $\exp(\M \tau)$ that costs $\BigO(D^5 \log(D)^2)$.
Additionally, it requires $K-1$ induction steps,
for each of which the definite integration
\[
	\int_{a_k}^{b_k} \exp(\M_{\Phi_k} (b_k-t_k)) \cdot \P_{\neg\Phi_{k-1}} \cdot \M_{\Phi_{k-1}} \cdot
	\exp(\M_{\Phi_{k-1}} (t_k-a_k)) \ d t_k
\]
dominates the computational cost.
The definite integral can be computed by at most $N^2-1$ times of integration-by-parts,
whose integrands can be obtained in $\BigO(N^6)$
as products of a few matrices with dimension $N^2$.
Hence the exactly calculating procedure is in
$\BigO(K (D^5 \log(D)^2+N^8))=\BigO(K \cdot D^5 \log(D)^2)$
that is bounded polynomially in the encoding size $\|\QC\|$ and linearly in the size $\|\phi\|$.
However, if we adopt the Riemann sum to do the approximate definite integration,
after setting a step length $\delta$ or equivalently the number of samples $M=(b_k-a_k)/\delta$,
we could get an approximate value
\[
	\frac{1}{M}\sum_{i=1}^M \exp(\M_{\Phi_k} (M-i) \delta) \cdot \P_{\neg\Phi_{k-1}}
	\cdot \M_{\Phi_{k-1}} \cdot \exp(\M_{\Phi_{k-1}} i \delta),
\]
each term is obtained in $\BigO(N^6)$ as a product of a few matrices with dimension $N^2$.
Then the approximately calculating procedure is
in $\BigO(K \cdot (\|\epsilon\|+\log(N)+M) \cdot N^6)$.

\begin{thm}
	The CSL formula is decidable over quantum CTMCs.
\end{thm}
\begin{proof}
	Using Jordan decomposition, matrix exponentiation and integration-by-parts,
	we first claim that
	the satisfaction probability can be expressed as a number in the explicit form
	\begin{equation}\label{eq:EP}
		p=\alpha_1 \e^{\lambda_1} + \alpha_2 \e^{\lambda_2} + \cdots + \alpha_m \e^{\lambda_m},
	\end{equation}
	where $\alpha_1,\ldots,\alpha_m$ are algebraic numbers
	and $\lambda_1,\ldots,\lambda_m$ are distinct algebraic numbers.
	We observe the facts:
	\begin{enumerate}
		\item The governing matrix $\M_\Phi$ appeared in Eq.~\eqref{eq:cal}
		(see its definition in Eq.~\eqref{eq:gov})
		takes algebraic numbers as entries,
		since the entries of $\HH_s$ and $\LL_{s_1,s_2}$ are algebraic. 
		\item The characteristic polynomial of $\M_\Phi$ is an $\mathbb{A}$-polynomial.
		The eigenvalues $\lambda_1,\ldots,\lambda_m$ of $\M_\Phi$ are algebraic,
		as they are roots of that $\mathbb{A}$-polynomial by Lemma~\ref{lem:closed}.
		Those eigenvalues will make up all exponents in Eq.~\eqref{eq:EP} afterwards.
		\item By Jordan decomposition,
		we have $\M_\Phi=\mathds{T} \cdot \mathds{J} \cdot \mathds{T}^{-1}$,
		where the Jordan canonical form $\mathds{J}$ has entries $0$, $1$, and those eigenvalues $\lambda_i$,
		and the transformation matrix $\mathds{T}$
		takes linear expressions with $\mathbb{A}$-coefficients over entries of $\mathds{J}$
		as its entries, which are algebraic too.
		\item Since $\exp(\M_\Phi\cdot a)=\mathds{T} \cdot \exp(\mathds{J}\cdot a) \cdot \mathds{T}^{-1}$,
		we have that the entries of $\exp(\mathds{J}\cdot a)$ are with
		the form $(\lambda_i a)^k \e^{\lambda_i a}/k!$ for some integer $k$.
		Thus the matrix exponential $\exp(\M_\Phi\cdot a)$
		takes linear expressions with $\mathbb{A}$-coefficients over $(\lambda_i a)^k \e^{\lambda_i a}/k!$
		as its entries.
		\item When $a$ is an algebraic constant that is taken from the endpoints of $\I_k$,
		the entries of $\exp(\M_\Phi\cdot a)$ are in the form~\eqref{eq:EP}.
		Whereas, $\exp(\M_\Phi\cdot t)$
		takes linear expressions with $\mathbb{A}$-coefficients over $(\lambda_i t)^k \e^{\lambda_i t}/k!$
		as its entries.
		\item The same structure holds for entries of the indefinite integral
		\[
			\int \exp(\M_{\Phi_k} (b_k-t_k)) \cdot
			\P_{\neg\Phi_{k-1}} \cdot \M_{\Phi_{k-1}} \cdot
			\exp(\M_{\Phi_{k-1}} (t_k-a_k)) \ d t_k,
		\]
		which follows integration-by-parts that for any polynomial $q(t)$,
		there is a polynomial $r(t)$ with $\deg(r) = \deg(q)$, such that
		$\int q(t) \e^{\lambda t} \ d t = r(t) \e^{\lambda t}$.
		Finally, we can see that the entries of the definite integral
		\[
			\int_{a_k}^{b_k} \exp(\M_{\Phi_k} (b_k-t_k)) \cdot
			\P_{\neg\Phi_{k-1}} \cdot \M_{\Phi_{k-1}} \cdot
			\exp(\M_{\Phi_{k-1}} (t_k-a_k)) \ d t_k
		\]
		are in the form~\eqref{eq:EP}.
	\end{enumerate}
	By Theorem~\ref{Lindemann},
	we have that $p=\textup{c}$ if and only if $m=1$, $\lambda_1=0$ and $\alpha_1=\textup{c}$.
	If not, we further compute the value $p$ up to any precision
	by sufficiently approaching the Euler constant $\e$,
	and thereby compare $p$ with $\textup{c}$,
	which decides the truth of $\QC \models \mathcal{P}_{\sim \textup{c}}(\phi)$.
\end{proof}

\paragraph{Complexity}
Although the decision algorithm is in time exponential
in the number of qubits involved in the quantum CTMC,
it is in time polynomial in the encoding size of that model
and linear in the size of the input CSL formula,
besides the query/oracle --- how sufficiently the Euler constant $\e$ should be approached.
Whereas, the satisfaction probability can be computed in polynomial time
using the numerical methods ---
the scaling and squaring method embedded with Pad\'{e} approximations and the Riemann sum.

\section{Concluding Remarks}\label{S5}
This paper introduced a novel model of quantum CTMC consisting of a classical subsystem and a quantum subsystem.
The essential difference of these two systems lies in that
a quantum system allows superpositions while a classical system does not.
The classical subsystem in the current model forbids superpositions,
making it more precise than the model of quantum CTMC proposed in~\cite{XMG+21},
which does not take this point into consideration.
Then we checked the well-known branching logic --- CSL --- against the quantum CTMC.
The decidability is established by an algebraic approach to tackling multiphase until formulas,
which lie at the core of CSL.
To be more efficient, numerical methods could be incorporated to yield a polynomial-time procedure
with respect to the dimension of the state space.
Furthermore,
we provided the Apollonian network as a running example to demonstrate the method mentioned above.

For future work, an interesting direction is
model-checking the quantum continuous-time Markov decision processes (CTMDP)~\cite{YiY18},
in which probabilistic and nondeterministic behaviors coexist.
To resolve nondeterminism,
following the verification work~\cite{Mil68,MSS20} in classical CTMDPs,
an optimal scheduler is expected to maximize the probability of a given property
over almost all time-stamps at each classical state.
Here, it follows that local optimization yields global optimization.
In quantum systems, however,
there are additionally quantum states that are described by density operators.
The known way to compare density operators $\rho$ are based on the L{\"o}wner partial order $\sqsubseteq$,
i.e., $\rho_1 \sqsubseteq \rho_2$ if $\rho_2-\rho_1$ is positive.
Thus local optimization does not yield global optimization generally,
which will bring nontrivial technical hardness to be resolved.

\section*{Acknowledgment}
\noindent The authors are grateful to the three anonymous reviewers
whose careful and insightful comments significantly improve the presentation and resolve inconsistencies.

  %% the following bibliography is gererated manually for the sake of brevity
  %% only; please use a separate .bib file in your submission

\bibliographystyle{alphaurl}
\bibliography{check}

\appendix
\section{Implementation}\label{A1}
The presented method has been implemented in the Wolfram Language on Mathematica 12.1
with Intel Core i7-10700 CPU at 2.90GHz.
We deliver it as user-friendly functions in the file \textsf{Functions.nb}
available at \url{https://github.com/meijingyi/CheckQCTMC}.
The main functions are listed below.
\begin{itemize}
	\item \texttt{ExpSS} computes the matrix exponential, as a subroutine from MATLAB,
	using the scaling and squaring method~\cite{AMH09}.	
	\item \texttt{GovernMat} computes the governing matrix $\M_\Phi$
	which keeps the behavior of $\Phi$-states and absorbs others.
	\item \texttt{Cyl} gets the final partial density operator in Eq.~\eqref{eq:inductive}
	while computing the probability measure along a given cylinder set.
	\item \texttt{Multiphase} gets the final partial density operator in Eq.~\eqref{eq:cal}
	while computing the probability measure of a multiphase until formula.
\end{itemize}
For the function \texttt{ExpSS}, there is a useful tool MATLink
available to transfer data between Mathematica and MATLAB,
thus we can use the embedded function \texttt{expm} in MATLAB.

After specifying a quantum CTMC model $\QC$ attached with the Hilbert space
and a multiphase until formula $\phi$,
one can perform the model-checking by calling the above functions respectively.
The running examples are validated in the file \textsf{Apollonian Network.nb}.
Notably, using the scaling and squaring method for matrix exponentiation,
the time and the space consumption of computing the probabilities of the cylinder set in Example~\ref{ex2}
and the multiphase until formula in Example~\ref{ex3} are summarized in Table~\ref{tab}.

\begin{table}[htbp]
	\setlength\tabcolsep{10pt}
	\renewcommand\arraystretch{1.2}
	\centering   
	\caption{Performance using the scaling and squaring method}\label{tab}
	\begin{tabular}{ccccccc}
		\toprule
		\multirow{2}{3.5em}{method} & \multicolumn{2}{c}{\texttt{Cyl}}
		& \multicolumn{2}{c}{\texttt{Multiphase}} & \multicolumn{2}{c}{\texttt{Overall}} \\
		& time & memory & time & memory & time & memory \\
		\midrule
		\texttt{ExpSS} & 0.44s & 4.60MB & 4.80s & 36.54MB & 5.375s & 176.531MB \\
		\bottomrule
	\end{tabular}
\end{table}

\end{document}